\documentclass[journal,twoside,web]{ieeecolor}
\usepackage{generic}

\usepackage{cite}
\usepackage{amsmath,amssymb,amsfonts}
\usepackage{graphicx}
\usepackage[hidelinks]{hyperref}
\usepackage{textcomp}

\usepackage{mathrsfs}
\usepackage{mathtools}
\usepackage{stmaryrd}
\usepackage{bm}
\usepackage{cite}
\usepackage{subfigure}
\usepackage{booktabs}
\usepackage{algorithm}
\usepackage[noend]{algpseudocode}
\usepackage{physics}

\algrenewcommand\algorithmicindent{1em}

\makeatletter
\renewcommand{\ALG@name}{Protocol}
\makeatother
\algrenewcommand\algorithmicdo{}

\newtheorem{definition}{Definition}
\newtheorem{assumption}{Assumption}

\newtheorem{lemma}{Lemma}
\newtheorem{theorem}{Theorem}
\newtheorem{corollary}{Corollary}

\newtheorem{remark}{Remark}

\newcommand{\defref}[1]{Definition~\ref{#1}}
\newcommand{\asmref}[1]{Assumption~\ref{#1}}

\newcommand{\lemref}[1]{Lemma~\ref{#1}}
\newcommand{\thmref}[1]{Theorem~\ref{#1}}
\newcommand{\corref}[1]{Corollary~\ref{#1}}

\newcommand{\secref}[1]{Section~\ref{#1}}
\newcommand{\appref}[1]{Appendix~\ref{#1}}
\newcommand{\figref}[1]{Fig.~\ref{#1}}
\newcommand{\tabref}[1]{Table~\ref{#1}}
\newcommand{\proref}[1]{Protocol~\ref{#1}}

\newcommand{\N}{\mathbb{N}}
\newcommand{\Z}{\mathbb{Z}}
\newcommand{\Q}{\mathbb{Q}}
\newcommand{\R}{\mathbb{R}}
\newcommand{\ZZ}[1]{\Z_{\langle #1 \rangle}}
\newcommand{\QQ}[2]{\Q_{\langle #1, #2 \rangle}}
\newcommand{\Share}{\mathsf{Share}}
\newcommand{\Reconst}{\mathsf{Reconst}}
\newcommand{\Mult}{\mathsf{Mult}}
\newcommand{\Trunc}{\mathsf{Trunc}}
\newcommand{\PRF}{\mathsf{PRF}}
\newcommand{\PCF}{\mathsf{PCF}}
\newcommand{\aux}{\mathsf{aux}}
\newcommand{\bfzero}{\mathbf{0}}

\DeclareMathOperator{\inv}{inv}

\DeclarePairedDelimiter{\floor}{\lfloor}{\rfloor}
\DeclarePairedDelimiter{\round}{\lceil}{\rfloor}
\DeclarePairedDelimiter{\share}{\llbracket}{\rrbracket}

\allowdisplaybreaks[1]
\def\BibTeX{{\rm B\kern-.05em{\sc i\kern-.025em b}\kern-.08em
    T\kern-.1667em\lower.7ex\hbox{E}\kern-.125emX}}
\begin{document}
\title{Client-Aided Secure Two-Party Computation of Dynamic Controllers}
\author{Kaoru Teranishi, \IEEEmembership{Member, IEEE} and Takashi Tanaka, \IEEEmembership{Senior Member, IEEE}
\thanks{Kaoru Teranishi and Takashi Tanaka are with the School of Aeronautics and Astronautics, Purdue University, West Lafayette, IN 47907-2045, USA (e-mail: kteranis@purdue.edu, tanaka16@purdue.edu).}
\thanks{Kaoru Teranishi is also with Japan Society for the Promotion of Science, Chiyoda, Tokyo 102-0083, Japan.}}

\thispagestyle{empty}
\hspace{-4.5mm}
\fbox{
\begin{minipage}{\textwidth-5mm}\scriptsize
© 20XX IEEE.
Personal use of this material is permitted.
Permission from IEEE must be obtained for all other uses, in any current or future media, including reprinting/republishing this material for advertising or promotional purposes, creating new collective works, for resale or redistribution to servers or lists, or reuse of any copyrighted component of this work in other works.
\end{minipage}
}
\newpage
\setcounter{page}{0}

\maketitle
\thispagestyle{empty}

\begin{abstract}
    In this paper, we propose a secure two-party computation protocol for dynamic controllers using a secret sharing scheme.
    The proposed protocol realizes outsourcing of controller computation to two servers, while controller parameters, states, inputs, and outputs are kept secret against the servers.
    Unlike previous encrypted controls in a single-server setting, the proposed method can operate a dynamic controller for an infinite time horizon without controller state decryption or input re-encryption.
    We show that the control performance achievable by the proposed protocol can be made arbitrarily close to that attained by the unencrypted controller.
    Furthermore, system-theoretic and cryptographic modifications of the protocol are presented to improve the communication complexity.
    The feasibility of the protocol is demonstrated through numerical examples of PID and observer-based controls.
\end{abstract}


\section{Introduction}
\label{sec:introduction}
\IEEEPARstart{T}{he} need for protecting privacy and security in control systems has become critical due to growing cyber-physical systems in various industries.
Encrypted control, which integrates automatic control with cryptographic technology, has emerged to prevent eavesdropping attacks and securely outsource controller computations for an untrusted third party~\cite{Darup2021-qq,Kim2022-ck,Schluter2023-ia}.
The initial realization of encrypted control was achieved in a single-server setting by using homomorphic encryption~\cite{Kogiso2015-go}.
Homomorphic encryption is an encryption scheme that allows arithmetic operations such as addition and multiplication on encrypted data~\cite{Regev2009-ys,Brakerski2014-ny,Fan2012-jt,Gentry2013-eq,Cheon2017-pd}.
In this framework, as illustrated in \figref{fig:controls}\subref{fig:single}, the server receives encrypted measurements from the sensor, performs controller computation using encrypted parameters without decryption, and returns encrypted control actions to the actuator.
Unlike other privacy mechanisms, such as differential privacy~\cite{Han2018-pm}, encrypted control offers reliable security built on cryptography without a tradeoff between control performance and security level~\cite{Schluter2023-ia}.

With existing homomorphic encryption schemes, arithmetic operations on encrypted data are usually performed on a finite range of integers.
One of the main challenges in this context is to prevent overflow in recursive state updates when encrypting dynamic controllers (see \secref{sec:dynamic_controllers_over_integers}).
To address the overflow issue, fully homomorphic encryption with bootstrapping was employed in~\cite{Kim2016-lc}, and the periodic reset of controller states was considered in~\cite{Murguia2020-ov}.
Recent approaches to this issue involve reformulating controller representations.
For instance, one approach approximated a controller using a finite impulse response filter~\cite{Schluter2021-ek}.
In addition, another approach reformulated a controller based on pole placement and coordinate transformation using re-encrypted control inputs~\cite{Kim2023-pk}.
The input re-encryption was also applied in~\cite{Teranishi2024-ha} and~\cite{Lee2025-jo}, which represented controller states as historical input and output data.

\begin{figure}[t]
    \centering
    \subfigure[Encrypted control in a single-server setting.]{\includegraphics[scale=1]{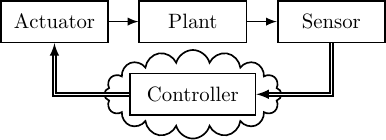}\label{fig:single}}
    \subfigure[Encrypted control in a two-party computation setting.]{\includegraphics[scale=1]{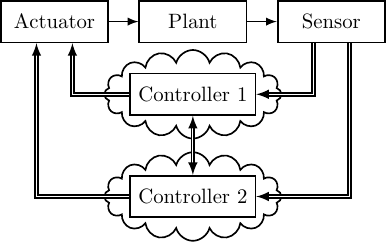}\label{fig:2pc}}
    \caption{System architectures of encrypted controls in single-server and two-party computation settings.}
    \label{fig:controls}
    \vspace{-5mm}
\end{figure}

While these methods have advanced encrypted control in a single-server setting, exploring encrypted control in a multi-party computation setting presents potential advantages.
Multi-party computation using secret sharing~\cite{Cramer2015-sj} is generally more computationally efficient and secure than homomorphic encryption in exchange for an increase in communication costs.
Furthermore, multi-party computation offers the ability to evaluate various functionalities.
This property motivates the implementation of encrypted control in a multi-party setting because homomorphic encryption typically requires substantial computational resources for computing mathematical operations beyond addition and multiplication.

\figref{fig:controls}\subref{fig:2pc} illustrates a typical system architecture of two-party encrypted control using secret sharing.
Each controller receives a share of sensor measurements, securely performs controller computation (along with communication to the other controller if needed), and returns a share of control inputs to the actuator.
In this framework, an efficient and secure implementation of static controllers was suggested~\cite{Darup2019-ou}.
This method was extended to polynomial control by combining secret sharing with homomorphic encryption~\cite{Darup2020-cm}.
The polynomial control was also realized based only on secret sharing in more than a three-party setting~\cite{Schlor2021-am}.
Furthermore, in a two-party setting, secret sharing with a garbled circuit was applied to implement a neural-network-based controller~\cite{Tjell2021-ar}, and model predictive control using homomorphic encryption was considered~\cite{Alexandru2018-jv}.

In this paper, we aim to design an efficient protocol for securely implementing dynamic controllers.
For this purpose, we follow the multi-party computation approach and perform bit truncation over secret sharing to prevent controller state overflow due to recursive multiplication.
The main contribution of this study is to demonstrate the effectiveness of multi-party computation for encrypted dynamic control.
Our contributions are summarized as follows:

\begin{itemize}
    \item
    We realize secure two-party computation of \emph{any} linear time-invariant controllers \emph{without requiring state decryption or input re-encryption}.
    This requirement of previous encrypted dynamic controls in a single-server setting~\cite{Kim2023-pk,Teranishi2024-ha,Lee2025-jo} is undesirable because of increasing computation and communication efforts on a client.
    Compared to the previous schemes, the proposed protocol is simpler and significantly reduces the computation costs of a client.

    \item
    We analyze the achievable control performance of the proposed protocol.
    This analysis clarifies that the performance can be arbitrarily close to that of the original controller by increasing the precision of the fixed-point numbers approximating the controller parameters and initial state.
    The result is examined by numerical examples with PID and observer-based controllers.

    \item
    The proposed protocol guarantees statistical security as long as the two parties do not collude with each other.
    It is stronger than computational security satisfied by conventional encrypted controls using homomorphic encryption.
\end{itemize}

One of the main drawbacks of the proposed protocol is a high communication cost when implementing large-scale controllers.
To attenuate communication efforts, we also provide modifications to the proposed protocol based on system-theoretic and cryptographic techniques.

The remainder of this paper is organized as follows.
\secref{sec:problem_setting} describes a problem setting.
\secref{sec:dynamic_controllers_over_integers} revisits the overflow problem in encrypting dynamic controllers.
\secref{sec:secret_sharing_and_beaver_triple} introduces secret sharing and Beaver triple.
\secref{sec:two-party_computation_of_dynamic_controllers} proposes a protocol for secure two-party computation of dynamic controllers and analyzes its control performance.
It also compares the computational and communication complexities of the proposed protocol with those of previous schemes.
\secref{sec:improvement_of_communication_complexities} presents system-theoretic and cryptographic improvements for reducing communication complexity.
\secref{sec:numerical_examples} examines the efficiency of the proposed protocol by numerical examples.
Finally, \secref{sec:conclusion} concludes the paper.

\emph{Notation:}
The sets of integers, nonnegative integers, positive integers, and real numbers are denoted by $\Z$, $\N_0$, $\N$, and $\R$, respectively.
For $q \ge 2$, $\Z_q \coloneqq \Z \cap [-q/2, q/2)$ denotes the set of integers within $[-q/2, q/2)$.
For $m \in \Z$, $m \bmod q \coloneqq m - \floor{ \frac{m + q / 2}{q} } q$ denotes the reduction of $m$ modulo $q$ over $\Z_q$.
For $m \in \Z_q$, $\inv(m, q) \in \Z_q$ denotes the modular inverse of $m$ modulo $q$, namely $\inv(m, q) \cdot m \bmod q = 1$.
For $x \in \R$, define the floor and rounding functions as $\floor{x} \coloneqq \max\{m \in \Z \mid m \le x\}$ and $\round{x} \coloneqq \floor{x + 1/2}$, respectively.
For vectors and matrices, $(\cdot \bmod q)$, $\floor{ \cdot }$, and $\round{ \cdot }$ are also used element wise.
The $n$-by-$n$ identity and $m$-by-$n$ zero matrices are denoted by $I_n$ and $\bfzero_{m \times n}$, respectively.
The $\ell_2$ and $\ell_\infty$ norms of vector $v \in \R^n$ are denoted by $\norm{v}$ and $\norm{v}_\infty$, respectively.
The induced $2$-norm and infinity norm of a matrix $M \in \R^{m \times n}$ are denoted by $\norm{M}$ and $\norm{M}_\infty$, respectively.

\section{Problem Setting}
\label{sec:problem_setting}

Consider the discrete-time linear time-invariant plant
\begin{equation}
    \begin{aligned}
        x_p(t + 1) &= A_p x_p(t) + B_p u(t), \\
        y(t) &= C_p x_p(t),
    \end{aligned}
    \label{eq:plant}
\end{equation}
where $t \in \N_0$ is the time index, $x_p(t) \in \R^{n_p}$ is the state, $u(t) \in \R^m$ is the input, $y(t) \in \R^p$ is the output, and the initial state is $x_p(0) = x_{p,0}$.
Suppose the plant is controlled by the discrete-time linear time-invariant controller
\begin{equation}
    \begin{aligned}
        x(t + 1) &= A x(t) + B y(t), \\
        u(t) &= C x(t) + D y(t),
    \end{aligned}
    \label{eq:controller}
\end{equation}
where $x(t) \in \R^n$ is the controller state, and its initial value is $x(0) = x_0$.

Define the set of $k$-bit fixed-point numbers with an $\ell$-bit fractional part as $\QQ{k}{\ell} \coloneqq \{ 2^{-\ell} z \mid z \in \ZZ{k} \}$, where $\ZZ{k} \coloneqq \Z_{2^k} = \{ -2^{k - 1}, \dots, 2^{k - 1} - 1 \}$ represents the set of $k$-bit integers.
We assume that the controller parameters and the initial state are represented by $k$-bit fixed-point numbers.

\begin{assumption}
\label{asm:fixed-point}
    Let $k$ and $\ell$ be positive integers, such that $k > \ell$.
    The controller parameters and initial state satisfy $A \in \QQ{k}{\ell}^{n \times n}$, $B \in \QQ{k}{\ell}^{n \times p}$, $C \in \QQ{k}{\ell}^{m \times n}$, $D \in \QQ{k}{\ell}^{m \times p}$, and $x_0 \in \QQ{k}{\ell}^n$.
\end{assumption}

Note that this assumption is reasonable in practice because fixed-point numbers can approximate real-valued controller parameters and initial state with desired precision $k$ and $\ell$.

We also assume that the closed-loop system consisting of \eqref{eq:plant} and \eqref{eq:controller},
\begin{equation}
    \begin{bmatrix}
        x_p(t + 1) \\
        x(t + 1)
    \end{bmatrix}
    = \Phi
    \begin{bmatrix}
        x_p(t) \\
        x(t)
    \end{bmatrix}, \ 
    \Phi \coloneqq
    \begin{bmatrix}
        A_p + B_p D C_p & B_p C \\
        B C_p & A
    \end{bmatrix},
    \label{eq:closed-loop}
\end{equation}
is stable.

\begin{assumption}
\label{asm:stability}
    The matrix $\Phi$ in \eqref{eq:closed-loop} is Schur stable.
\end{assumption}

The goal of this study is to design a secure two-party protocol for running a dynamic controller without state decryption and input re-encryption, where the control performance with the protocol is sufficiently close to that of the original controller \eqref{eq:controller}.
In other words, we will propose a secure protocol that outputs $\hat{u}(t)$ satisfying, for any given $\epsilon > 0$,
\begin{equation}
    \norm*{u(t) - \hat{u}(t)} \le \epsilon
    \label{eq:error}
\end{equation}
for all $t \in \N_0$.

We adopt a client-aided computation model~\cite{Ohata2020-py} for two-party computation.
In this setting, the client who initiates queries to two parties provides not only the protocol inputs but also auxiliary inputs, although it does not interact with them during the computation.
The parties use these auxiliary inputs to execute the protocol efficiently and return the outputs to the client.
We also assume peer-to-peer secure channels as a network model, which can be realized by standard symmetric-key encryption.
Furthermore, throughout this paper, we assume the following standard assumptions for two-party computation.

\begin{assumption}
\label{asm:adversary}
    The parties who perform a protocol are semi-honest and do not collude with each other.
\end{assumption}

This assumption implies that the parties follow our intended protocol but independently attempt to learn private information of the client from their received messages.
It also prohibits external adversaries from corrupting or accessing both parties.

\section{Dynamic Controllers Over Integers}
\label{sec:dynamic_controllers_over_integers}

Our basic idea for designing a protocol for the secure implementation of the dynamic controller \eqref{eq:controller} is to split it into two randomized controllers.
The parameters and states of these controllers are represented by the shares of a secret sharing scheme.
Each randomized controller is installed on a distinct server that computes a share of $u(t)$ using the controller with a share of $y(t)$ and returns it to the plant.

To perform controller computation with secret sharing, the controller parameters, states, inputs, and outputs should be encoded into matrices and vectors over integers.
This section revisits the well-known overflow problem~\cite{Cheon2018-vr,Kim2023-pk} caused by such encoding and clarifies the nature of the difficulty of encrypting dynamic controllers using homomorphic encryption.

We begin by encoding the controller parameters, the initial state, and the plant output.
\asmref{asm:fixed-point} results in the encoded parameters and initial state given as
\begin{equation}
    \begin{alignedat}{3}
        \bar{A} &\coloneqq 2^\ell A \in \ZZ{k}^{n \times n}, &\ \bar{B} &\coloneqq 2^\ell B \in \ZZ{k}^{n \times p}, &\ \bar{C} &\coloneqq 2^\ell C \in \ZZ{k}^{m \times n}, \\
        \bar{D} &\coloneqq 2^\ell D \in \ZZ{k}^{m \times p}, &\ \bar{x}_0 &\coloneqq 2^\ell x_0 \in \ZZ{k}^n. & &
    \end{alignedat}
    \label{eq:encoded_parameters}
\end{equation}
Similarly, the output of the plant \eqref{eq:plant} can be encoded as
\[
    \bar{y}(t) \coloneqq \round*{2^\ell y(t)} \in \Z^p,
\]
where the rounding is performed in an element-wise manner.

With the encoded parameters, state, and output, one might consider the following controller over integers~\cite{Kim2016-lc,Murguia2020-ov},
\begin{align*}
    z(t + 1) &= \bar{A} z(t) + 2^{t \ell} \bar{B} \bar{y}(t), \\
    v(t) &= \bar{C} z(t) + 2^{t \ell} \bar{D} \bar{y}(t),
\end{align*}
where $z(t) \in \Z^n$, $v(t) \in \Z^m$, and $z(0) = \bar{x}_0$.
Unfortunately, this controller is problematic due to the increase in the bit length of the controller state.
The controller state at time $t \in \N$ is given as
\[
    z(t) \approx 2^{(t + 1) \ell} \qty( A^t x_0 + \sum_{s=0}^{t-1} A^{t - 1 -s} B y(s) ) = 2^{(t + 1) \ell} x(t),
\]
if $\bar{y}(t) \approx 2^\ell y(t)$.
Hence, the number of bits in the controller state increases by $\ell$ at each time step.
This implies that the norm of the encoded state $z(t)$ eventually diverges as time $t$ approaches infinity because $\ell > 0$.
Therefore, the encoded state causes overflow at a finite time step.

To avoid this problem, one can consider another controller with truncation~\cite{Kim2023-pk} given as
\begin{equation}
    \begin{aligned}
        \bar{x}(t + 1) &= \round*{ \qty( \bar{A} \bar{x}(t) + \bar{B} \bar{y}(t) ) / 2^\ell }, \\
        \bar{u}(t) &= \bar{C} \bar{x}(t) + \bar{D} \bar{y}(t),
    \end{aligned}
    \label{eq:encoded_controller}
\end{equation}
where $\bar{x}(t) \in \Z^n$ and $\bar{u}(t) \in \Z^m$ denote the encoded state and input, respectively, and $\bar{x}(0) = \bar{x}_0$.
In this controller, the least $\ell$ bits of the state update $\bar{A} \bar{x}(t) + \bar{B} \bar{y}(t)$ are truncated by $\round{ \cdot / 2^\ell }$.
Then, the bit length of the controller state does not increase because
\[
    \bar{x}(t + 1) \approx 2^\ell \qty( A x(t) + B y(t) ) = 2^\ell x(t + 1),
\]
for a sufficiently large $\ell$.
Additionally, in this case, the control input can be recovered using $u(t) \approx 2^{-2 \ell} \bar{u}(t)$.

This approximation suggests that truncation solves the problem of controller state overflow, making \eqref{eq:encoded_controller} suitable for encoding the controller \eqref{eq:controller}.
However, implementing the encoded controller using homomorphic encryption remains problematic.
The difficulty lies in the fact that homomorphic encryption is supposed to be unable to compute the truncation $\round{ \cdot / 2^\ell }$ on encrypted data directly~\cite{Kim2023-pk}.
This is a fundamental problem in the encryption of dynamic controllers using homomorphic encryption.
We address this problem by constructing an efficient protocol to compute the truncation using secret sharing, further discussed in~\secref{sec:two-party_computation_of_dynamic_controllers}.

\section{Secret Sharing and Beaver Triple}
\label{sec:secret_sharing_and_beaver_triple}

This section provides a brief review of secret sharing and its arithmetic~\cite{Cramer2015-sj,Beaver2007-kr}.

\subsection{Secret Sharing}

A $t$-out-of-$n$ secret sharing scheme splits a secret into $n$ shares, and the share can be recovered by collecting $t$ shares.
In this study, we use a standard $2$-out-of-$2$ secret sharing scheme over $\Z_q$ with a prime $q$ defined as follows.

\begin{definition}[$2$-out-of-$2$ secret sharing]
    A $2$-out-of-$2$ secret sharing scheme over $\Z_q$ is a tuple of polynomial-time algorithms $\Share$ and $\Reconst$ such that:
    \begin{itemize}
        \item $\share{m} \gets \Share(m)$:
        The share generation algorithm takes a secret $m \in \Z_q$ as input, randomly choose $r \in \Z_q$, and outputs shares $\share{m} = (\share{m}_1, \share{m}_2) = (r, m - r \bmod q)$.
        
        \item $m \gets \Reconst(\share{m})$:
        The reconstruction algorithm takes shares $\share{m} = (\share{m}_1, \share{m}_2)$ as input and outputs $m = \share{m}_1 + \share{m}_2 \bmod q$.
    \end{itemize}
\end{definition}

Each share $\share{m}_i$, $i \in \{1, 2\}$, is uniformly distributed over $\Z_q$ and thus provides no information about $m$ to the party $P_i$.
Nevertheless, any secret in $\Z_q$ can be correctly recovered by collecting the shares, i.e.,
\begin{equation}
    \Reconst(\Share(m)) = m
    \tag*{(P1)}
    \label{eq:correctness}
\end{equation}
for all $m \in \Z_q$.
Furthermore, the secret sharing scheme allows the constant addition, constant multiplication, and addition of shares.
For $k, x \in \Z_q$, the constant addition and multiplication of $k$ and $\share{x}$ are defined as $\share{x} + k \coloneqq (\share{x}_1 + k \bmod q, \share{x}_2)$ and $k \share{x} \coloneqq (k \share{x}_1 \bmod q, k \share{x}_2 \bmod q)$, respectively.
The addition of $\share{x}$ and $\share{y}$ with $y \in \Z_q$ is also defined as $\share{x} + \share{y} \coloneqq (\share{x}_1 + \share{y}_1 \bmod q, \share{x}_2 + \share{y}_2 \bmod q)$.
By definition, it holds that
\begin{align*}
    \share{x} + k &= \share{x + k \bmod q}, \tag*{(P2)} \label{eq:constant_addition} \\
    k \share{x} &= \share{k x \bmod q}, \tag*{(P3)} \label{eq:constant_multiplication} \\
    \share{x} + \share{y} &= \share{x + y \bmod q}. \tag*{(P4)} \label{eq:addition}
\end{align*}
We also define the subtractions, $\share{x} - k$ and $\share{x} - \share{y}$, similar to the additions.
The above computations can be performed without interaction, meaning that $x$ and $y$ are kept secret from both the parties.

\subsection{Beaver Triple}

Unlike addition and subtraction, the multiplication of shares cannot be naturally defined.
That is, $(\share{x}_1 \share{y}_1 \bmod q, \share{x}_2 \share{y}_2 \bmod q) \ne \share{xy \bmod q}$ because $xy \bmod q = (\share{x}_1 + \share{x}_2) (\share{y}_1 + \share{y}_2) \bmod q = \share{x}_1 \share{y}_1 + \share{x}_1 \share{y}_2 + \share{x}_2 \share{y}_1 + \share{x}_2 \share{y}_2 \bmod q$.
Although the terms $\share{x}_1 \share{y}_1$ and $\share{x}_2 \share{y}_2$ can be computed locally by each party, the computations of $\share{x}_1 \share{y}_2$ and $\share{x}_2 \share{y}_1$ require the exchange of shares, thereby compromising privacy.

To compute valid shares of $xy \bmod q$ privately, we use a technique based on Beaver triples~\cite{Beaver2007-kr}.
A Beaver triple over $\Z_q$ is a triplet $(\share{a}, \share{b}, \share{c})$ that consists of shares of random numbers $a, b, c \in \Z_q$ such that $c = ab \bmod q$.
With the triplet, the multiplication of shares is computed as follows~\cite{Beaver2007-kr,Cramer2015-sj}.
Suppose that the parties $P_i$, $i \in \{1, 2\}$ have $\share{x}_i$, $\share{y}_i$, and $(\share{a}_i, \share{b}_i, \share{c}_i)$.
Each party computes $\share{d}_i = \share{x}_i - \share{a}_i \bmod q$ and $\share{e}_i = \share{y}_i - \share{b}_i \bmod q$ locally.
The parties open $\share{d}_i$ and $\share{e}_i$ and compute $d = \Reconst(\share{d})$ and $e = \Reconst(\share{e})$, where $\share{d} = (\share{d}_1, \share{d}_2)$ and $\share{e} = (\share{e}_1, \share{e}_2)$.
Then, $P_1$ and $P_2$ respectively compute $\share{z}_1 = e \share{a}_1 + d \share{b}_1 + \share{c}_1 + d e \bmod q$ and $\share{z}_2 = e \share{a}_2 + d \share{b}_2 + \share{c}_2 \bmod q$ locally.
The shares $\share{z} = (\share{z}_1, \share{z}_2)$ satisfy $\share{z} = \share{xy \bmod q}$ because
\begin{align*}
    &\Reconst(\share{z}) = \share{z}_1 + \share{z}_2 \bmod q, \\
    &= e ( \share{a}_1 + \share{a}_2 ) + d ( \share{b}_1 + \share{b}_2 ) + ( \share{c}_1 + \share{c}_2 ) \\
    &\phantom{{}={}} + d e \bmod q, \\
    &= (y - b) a + (x - a) b + ab + (x - a) (y - b) \bmod q, \\
    &= xy \bmod q,
\end{align*}
where we used the properties from \refeq{eq:correctness} to \refeq{eq:addition}.

\proref{pro:mult} summarizes the above procedure for computing the multiplication of shares.
Using the $\Mult$ protocol, we define
\begin{equation}
    \share{x} \share{y} \coloneqq \Mult(\share{x}, \share{y}) = \share{xy \bmod q}.
    \tag*{(P5)}
    \label{eq:multiplication}
\end{equation}
Note that the protocol does not leak any information about $x$ and $y$ to the parties because of the randomness of $d = x - a \bmod q$ and $e = y - b \bmod q$.

\begin{figure}[t]
    \begin{algorithm}[H]
        \caption{Multiplication ($\Mult$)}
        \label{pro:mult}
        \begin{algorithmic}[1]
            \Require Shares $\share{x}$ and $\share{y}$ over $\Z_q$ and a Beaver triple $(\share{a}, \share{b}, \share{c})$ over $\Z_q$.
            \Ensure Shares $\share{z}$ over $\Z_q$ such that $\share{z} = \share{xy \bmod q}$.
            \State $P_i$: $\share{d} \gets \share{x} - \share{a}$, $\share{e} \gets \share{y} - \share{b}$
            \State $P_i$: $d \gets \Reconst(\share{d})$, $e \gets \Reconst(\share{e})$
            \State $P_i$: $\share{z} \gets e \share{a} + d \share{b} + \share{c} + d e$
        \end{algorithmic}
    \end{algorithm}
    \vspace{-8mm}
\end{figure}

For $M \in \Z_q^{d_1 \times d_2}$, denote $\share{M} \gets \Share(M)$ as a matrix in which the shares of each element of $M$ are arranged in the same manner, namely
\[
    \share{M} =
    \begin{bmatrix}
        \share{M_{11}}    & \cdots & \share{M_{1 d_2}} \\
        \vdots            & \ddots & \vdots \\
        \share{M_{d_1 1}} & \cdots & \share{M_{d_1 d_2}} \\
    \end{bmatrix}.
\]
Using the notation and properties \refeq{eq:addition} and \refeq{eq:multiplication}, the valid shares of matrix addition $X + Y$ for $X, Y \in \Z_q^{d_1 \times d_2}$ and multiplication $XZ$ for $X$ and $Z \in \Z_q^{d_2 \times d_3}$ can be computed as $\share{X} + \share{Y}$ and $\share{X} \share{Z}$, respectively.
The $(i, j)$ elements of the resultant matrices are $\qty( \share{X} + \share{Y} )_{ij} = \share{X_{ij} + Y_{ij} \bmod q}$ and $\qty( \share{X} \share{Z} )_{ij} = \share{X_{i1} Z_{1j} + \cdots + X_{i d_2} Z_{d_2 j} \bmod q}$.
Here, the matrix multiplication requires distinct $d_1 d_2 d_3$ Beaver triples.
In what follows, the $i$th shares of $\share{M}$ are denoted as
\[
    \share{M}_i =
    \begin{bmatrix}
        \share{M_{11}}_i    & \cdots & \share{M_{1 d_2}}_i \\
        \vdots              & \ddots & \vdots \\
        \share{M_{d_1 1}}_i & \cdots & \share{M_{d_1 d_2}}_i \\
    \end{bmatrix}
\]
for simplicity of notation.
In addition, the reconstruction algorithm $\Reconst$ is supposed to be performed for each element of the matrix, namely $\Reconst(\share{M}) = M$.

\begin{remark}
    Under \asmref{asm:adversary}, the $2$-out-of-$2$ secret sharing scheme satisfies perfect security~\cite{Shannon1949-mc}.
    This security means that parties can learn nothing from their own shares, even if their computational capabilities are unlimited.
\end{remark}

\section{Two-Party Computation of Dynamic Controllers}
\label{sec:two-party_computation_of_dynamic_controllers}

In this section, we propose an implementation of the encoded controller \eqref{eq:encoded_controller} over the two-party configuration shown in \figref{fig:controls}\subref{fig:2pc} using the $2$-out-of-$2$ secret sharing scheme.
The computation of \eqref{eq:encoded_controller} involves addition, multiplication, and truncation.
The addition and multiplication of shares are supported by the arithmetic of the secret sharing scheme, \refeq{eq:addition} and \refeq{eq:multiplication}, introduced in the previous section.
However, the truncation of shares cannot be computed directly from \refeq{eq:constant_addition} to \refeq{eq:multiplication}.

To enable this, we first introduce a truncation protocol originally proposed in \cite{Escudero2020-ih} and slightly modified here to suit our setting, along with its corresponding security notion.
We then apply the protocol to realize \eqref{eq:encoded_controller} over shares and analyze the control performance achievable by the proposed protocol.

\subsection{Truncation Protocol}

An important insight behind the truncation protocol in~\cite{Escudero2020-ih} is that, by definition of the modulo operation, the truncation of $\ell$ bits of $m \in \Z$ can be represented by
\begin{align*}
    \round*{ \frac{ m }{ 2^\ell } }
    &= \frac{ m - ( m \bmod 2^\ell ) }{ 2^\ell }, \\
    &= \inv(2^\ell, q) ( m - ( m \bmod 2^\ell ) ) \bmod q,
\end{align*}
if $q$ is a prime such that $q / 2 > \abs*{ \round{ m / 2^\ell } }$ and $q / 2 > 2^\ell$.
Notice that $\inv(2^\ell, q)$ is the modular inverse of $2^\ell$ modulo $q$ and can be computed using the extended Euclidean algorithm.
Here, let $q$ be the prime modulus of the secret sharing scheme in the previous section, and $\lambda \in \N$ be an appropriately chosen security parameter. 
Suppose that the bit length of a message $m$ is at most $\kappa = \floor{ \log_2 q } - \lambda - 1$ and $\ell$ is less than $\kappa$.
To simplify our protocol, we assume that $\ell$ is public, which is acceptable because it specifies only the bit length of fractional part.
Based on the above observation, the $\Trunc$ protocol shown in \proref{pro:truncation} computes the truncation of $\ell$ bits of a message $m \in \ZZ{\kappa}$ over shares.

\begin{figure}[t]
    \begin{algorithm}[H]
        \caption{Truncation ($\Trunc$)}
        \label{pro:truncation}
        \begin{algorithmic}[1]
            \Require Shares $\share{m}$ over $\Z_q$ with $m \in \ZZ{\kappa}$, shares $(\share{r}, \share{r'})$ over $\Z_q$ with random numbers $r \in \ZZ{\kappa - \ell + \lambda}$ and $r' \in \ZZ{\ell}$, and bit length $\ell \in \N$, where $\kappa = \floor{ \log_2 q } - \lambda - 1 > \ell$, and $\lambda \in \N$ is a security parameter.
            \Ensure Shares $\share{m'}$ over $\Z_q$ such that $m' = \round{ m / 2^\ell } + w$, where $w \in \{-1, 0, 1\}$.
            \State $P_i$: $\share{m_r} \gets \share{m} + 2^\ell \share{r} + \share{r'} + 2^{\ell - 1}$
            \State $P_2$: Send $\share{m_r}_2$ to $P_1$
            \State $P_1$: $m_r \gets \Reconst(\share{m_r})$
            \State $P_i$: $\share{m'} \gets \inv(2^\ell, q) ( \share{m} + \share{r'} - ( m_r - 2^{\ell - 1} \bmod 2^\ell ) )$
        \end{algorithmic}
    \end{algorithm}
    \vspace{-8mm}
\end{figure}

The protocol takes as input not only $\share{m} = (\share{m}_1, \share{m}_2)$ and $\ell$ but also shares $\share{r} = (\share{r}_1, \share{r}_2)$ and $\share{r'} = (\share{r'}_1, \share{r'}_2)$ with random numbers $r \in \ZZ{\kappa - \ell + \lambda}$ and $r' \in \ZZ{\ell}$.
Suppose that parties $P_i$, $i \in \{1, 2\}$ have $\share{m}_i$, $\share{r}_i$, and $\share{r'}_i$.
The parties compute $\share{m_r} = \share{m} + 2^\ell \share{r} + \share{r'} + 2^{\ell - 1}$ locally.
The party $P_2$ transmits $\share{m_r}_2$ to $P_1$, and $P_1$ reconstructs $m_r = \Reconst(\share{m_r})$, where $\share{m_r} = (\share{m_r}_1, \share{m_r}_2)$.
Then, $P_1$ and $P_2$ locally compute $\share{m'}_1 = \inv(2^\ell, q) ( \share{m}_1 + \share{r'}_1 - (m_r - 2^{\ell - 1} \bmod 2^\ell) ) \bmod q$ and $\share{m'}_2 = \inv(2^\ell, q) ( \share{m}_2 + \share{r'}_2 ) \bmod q$, respectively.
Consequently, $\share{m'} = (\share{m'}_1, \share{m'}_2)$ becomes shares of $\round{ m / 2^\ell } + w$ for all $m \in \ZZ{\kappa}$, where $w \in \{-1, 0, 1\}$.

\begin{lemma}
\label{lem:truncation}
    Suppose $q$ is a prime.
    Let $\kappa, \ell, \lambda \in \N$ such that $\kappa = \floor{ \log_2 q } - \lambda - 1 > \ell$.
    It holds that, for all $m \in \ZZ{\kappa}$,
    \[
        \Reconst \qty( \Trunc \qty( \share{ m }, \ell ) ) = \round*{ \frac{ m }{ 2^\ell } } + w,
    \]
    where $\share{m} \gets \Share(m)$ and $w \in \{-1, 0, 1\}$.
\end{lemma}

\begin{proof}
    See \appref{app:proof_truncation}.
\end{proof}

The lemma shows that the $\Trunc$ protocol truncates the least $\ell$ bits of $m \in \ZZ{\kappa}$ with a $1$-bit error $w$.
This error occurs only when $m + r' \notin \ZZ{\ell}$ and thus can be removed by an additional protocol to detect $m + r' \notin \ZZ{\ell}$ over shares~\cite{Escudero2020-ih}.
The following section tolerates the error to avoid increasing the computational and communication costs and applies the truncation protocol to approximate and implement the encoded controller \eqref{eq:encoded_controller} to run it over shares.

\begin{remark}
    The security of $\Trunc$ is identical to that of the truncation protocol in~\cite{Escudero2020-ih}.
    Through protocol execution, $P_2$ does not receive any messages from $P_1$ and thus learns no information.
    The only information leaked into $P_1$ is $m_r = m + 2^\ell r + r' + 2^{\ell - 1} \bmod q$.
    Under the conditions of bit lengths of $m$, $r$, and $r'$, it holds that $m_r = m + r'' \in \Z_q$, where $r'' = 2^\ell r + r' + 2^{\ell - 1}$ is uniformly distributed over $\ZZ{\kappa + \lambda}$.
    Intuitively, $m_r$ is a masked message of $m$ by $r''$, and thus $P_1$ cannot learn any information from $m_r$.
    Indeed, the statistical distance between $m + r''$ and $r''$, $\frac{ 1 }{ 2 } \sum_{x \in \Z_q} \abs*{ \Pr[ m + r'' = x ] - \Pr[ r'' = x ] }$, is less than $2^{-\lambda}$ and so is negligible in $\lambda$.
    Therefore, the $\Trunc$ protocol achieves statistical security, which is a relaxation of the perfect security.
    However, statistical security is still valid for unconditional adversaries.
    This implies that it is stronger than computational security achieved by homomorphic encryption~\cite{Regev2009-ys,Brakerski2014-ny,Fan2012-jt,Gentry2013-eq,Cheon2017-pd}.
    Note that computational security covers only conditional adversaries whose computational capability is limited to polynomial time~\cite{Katz2014-kb}.
\end{remark}

\subsection{Controller Computation Protocol}

With \proref{pro:mult} and \proref{pro:truncation}, we propose \proref{pro:controller}, which enables running a dynamic controller over the secret sharing scheme without state decryption and input re-encryption.
For convenience, we define the plant controlled by the proposed protocol instead of the controller \eqref{eq:controller} as
\begin{equation}
    \begin{aligned}
        \hat{x}_p(t + 1) &= A_p \hat{x}_p(t) + B_p \hat{u}(t), \\
        \hat{y}(t) &= C_p \hat{x}_p(t),
    \end{aligned}
    \label{eq:plant_2}
\end{equation}
where $\hat{x}_p(t) \in \R^{n_p}$ is the state, $\hat{u}(t) \in \R^m$ is the input provided by the protocol, $\hat{y}(t) \in \R^p$ is the output, and $\hat{x}_p(0) = x_p(0) = x_{p,0}$ is the initial state.
In that case, the encoded plant output is redefined as
\begin{equation}
    \bar{y}(t) \coloneqq \round*{2^\ell \hat{y}(t)} \in \Z^p.
    \label{eq:encoded_output}
\end{equation}
In what follows, we refer to the plant \eqref{eq:plant_2} with a sensor and actuator as a client and two servers computing a dynamic controller as parties $P_1$ and $P_2$.

\begin{figure}[t]
    \begin{algorithm}[H]
        \caption{Client-aided two-party controller computation}
        \label{pro:controller}
        \begin{algorithmic}[1]
            \Require Controller parameters $A \in \QQ{k}{\ell}^{n \times n}$, $B \in \QQ{k}{\ell}^{n \times p}$, $C \in \QQ{k}{\ell}^{m \times n}$, $D \in \QQ{k}{\ell}^{m \times p}$, initial state $x_0 \in \QQ{k}{\ell}^n$, plant output $\hat{y}(t) \in \R^p$, bit lengths $k, \ell \in \N$, and modulus $q$
            \Ensure Control input $\hat{u}(t) \in \R^m$
            \State $\triangleright$ Offline
            \State C: Compute $\share[\big]{\tilde{A}}$, $\share[\big]{\tilde{B}}$, $\share[\big]{\tilde{C}}$, $\share[\big]{\tilde{D}}$, and $\share[\big]{\tilde{x}(0)}$ by \eqref{eq:parameter_share}
            \State C: Send $\share[\big]{\tilde{A}}_i$, $\share[\big]{\tilde{B}}_i$, $\share[\big]{\tilde{C}}_i$, $\share[\big]{\tilde{D}}_i$, and $\share[\big]{\tilde{x}(0)}_i$ to $P_i$
            \State $\triangleright$ Online
            \ForAll{$t \in \N_0$}
                \State C: Compute $\share{\tilde{y}(t)}$ by \eqref{eq:output_share}
                \State C: Generate auxiliary inputs $\aux_1$ and $\aux_2$
                \State C: Send $\share{\tilde{y}(t)}_i$ and $\aux_i$ to $P_i$
                \State $P_i$: Compute $\share{\tilde{x}(t + 1)}$ and $\share{\tilde{u}(t)}$ by \eqref{eq:2pc_controller}
                \State $P_i$: Send $\share{\tilde{u}(t)}_i$ to the client
                \State C: Compute $\hat{u}(t)$ by \eqref{eq:approximate_input}
            \EndFor
        \end{algorithmic}
    \end{algorithm}
    \vspace{-8mm}
\end{figure}

To implement \eqref{eq:encoded_controller} with the secret sharing scheme, the client splits its controller parameters and initial state into shares using the share generation algorithm,
\begin{equation}
    \begin{alignedat}{3}
        \share[\big]{\tilde{A}} &\gets \Share \qty\big(\tilde{A}), &\ \share[\big]{\tilde{B}} &\gets \Share \qty\big(\tilde{B}), &\ \share[\big]{\tilde{C}} &\gets \Share \qty\big(\tilde{C}), \\
        \share[\big]{\tilde{D}} &\gets \Share \qty\big(\tilde{D}), &\ \share[\big]{\tilde{x}(0)} &\gets \Share \qty\big(\tilde{x}_0), & &
    \end{alignedat}
    \label{eq:parameter_share}
\end{equation}
where
\[
    \begin{alignedat}{3}
        \tilde{A} &\coloneqq \bar{A} \bmod q, &\quad \tilde{B} &\coloneqq \bar{B} \bmod q, &\quad \tilde{C} &\coloneqq \bar{C} \bmod q, \\
        \tilde{D} &\coloneqq \bar{D} \bmod q, &\quad \tilde{x}_0 &\coloneqq \bar{x}_0 \bmod q, & &
    \end{alignedat}
\]
and $\bar{A}$, $\bar{B}$, $\bar{C}$, $\bar{D}$, and $\bar{x}_0$ are the encoded controller parameters and initial state defined in \eqref{eq:encoded_parameters}. 
In the offline phase, as shown in \figref{fig:system}\subref{fig:offline}, the client sends the $i$th shares of the parameters $\share[\big]{\tilde{A}}_i, \share[\big]{\tilde{B}}_i, \share[\big]{\tilde{C}}_i, \share[\big]{\tilde{D}}_i$ and initial state $\share[\big]{\tilde{x}(0)}_i$ to the party $P_i$.
The parties store the received shares individually and use them in the online phase later.

In the online phase, the parties run an approximate dynamic controller for \eqref{eq:encoded_controller} as shown in \figref{fig:system}\subref{fig:online}.
For each time $t \in \N_0$, the client reads the plant output $\hat{y}(t)$ in \eqref{eq:plant_2} and generates its shares by
\begin{equation}
    \share{\tilde{y}(t)} \gets \Share \qty(\tilde{y}(t)), \quad \tilde{y}(t) = \bar{y}(t) \bmod q \in \Z_q^p,
    \label{eq:output_share}
\end{equation}
where $\bar{y}(t)$ is the encoded output defined in \eqref{eq:encoded_output}.
The $i$th shares of the output $\share{\tilde{y}(t)}_i$ are sent to the party $P_i$.
Simultaneously, the client generates $(n + m) (n + p)$ Beaver triples $(\share{a_j}, \share{b_j}, \share{c_j})$ over $\Z_q$ such that $c_j = a_j b_j \bmod q$ for every $j = 1, \dots, (n + m) (n + p)$ and shares $(\share{r_h}, \share{r'_h})$ over $\Z_q$ with random numbers $r_h \in \ZZ{\kappa - \ell + \lambda}$ and $r'_h \in \ZZ{\ell}$ for every $h = 1, \dots, n$.
The client then sends the $i$th shares
\begin{equation}
    \begin{aligned}
    \aux_i
    &\coloneqq \{ ( \share{a_j}_i, \share{b_j}_i, \share{c_j}_i ), ( \share{r_h}_i, \share{r'_h}_i ) \\
    &\phantom{{}\coloneqq{}} \mid j = 1, \dots, (n + m) (n + p), h = 1, \dots, n \}
    \end{aligned}
    \label{eq:auxiliary_inputs}
\end{equation}
to $P_i$ as auxiliary inputs.
The parties compute the shares of state update $\tilde{x}(t + 1) \in \Z^n$ and encoded input $\tilde{u}(t) \in \Z^m$ as
\begin{equation}
    \begin{aligned}
        \share[\big]{\tilde{x}(t + 1)} &\gets \Trunc \qty( \share[\big]{\tilde{A}} \share[\big]{\tilde{x}(t)} + \share[\big]{\tilde{B}} \share[\big]{\tilde{y}(t)}, \ell ), \\
        \share[\big]{\tilde{u}(t)} &= \share[\big]{\tilde{C}} \share[\big]{\tilde{x}(t)} + \share[\big]{\tilde{D}} \share[\big]{\tilde{y}(t)},
    \end{aligned}
    \label{eq:2pc_controller}
\end{equation}
where the $\Trunc$ protocol is applied for each element of $\share[\big]{\tilde{A}} \share[\big]{\tilde{x}(t)} + \share[\big]{\tilde{B}} \share[\big]{\tilde{y}(t)}$.
Note that the multiplications of shares for the state and input computations consume $n (n + p)$ and $m (n + p)$ Beaver triples, respectively.
Each party $P_i$ returns its shares of the control input $\share{\tilde{u}(t)}_i$ to the client separately.
Upon receiving the shares, the client reconstructs and decodes an approximate control input by
\begin{equation}
    \hat{u}(t) = 2^{-2 \ell} \tilde{u}(t), \quad \tilde{u}(t) = \Reconst \qty(\share{\tilde{u}(t)}).
    \label{eq:approximate_input}
\end{equation}

\begin{figure}[t]
    \centering
    \subfigure[Offline phase.]{\includegraphics[scale=1]{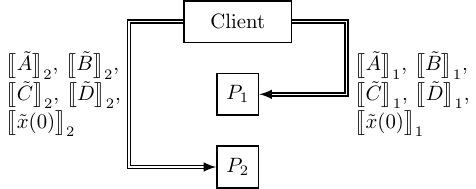}\label{fig:offline}}
    \subfigure[Online phase.]{\includegraphics[scale=1]{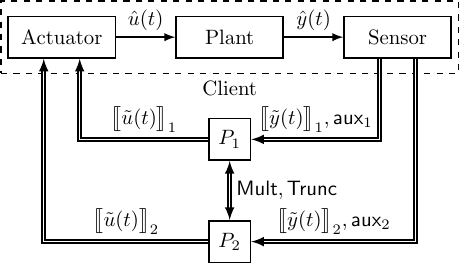}\label{fig:online}}
    \caption{Client-aided two-party dynamic controller computation.}
    \label{fig:system}
    \vspace{-5mm}
\end{figure}

We now proceed to show the main results of this study.
The lemma below guarantees that the proposed controller \eqref{eq:2pc_controller} can operate for an infinite time horizon without causing state overflow and does not require state decryption or input re-encryption.

\begin{lemma}
\label{lem:modulus_condition}
    Consider the plant \eqref{eq:plant_2} and controller \eqref{eq:2pc_controller} under \asmref{asm:fixed-point} and \asmref{asm:stability}.
    Define $\Gamma \coloneqq [ (B_p D)^\top \ B^\top ]^\top$, $\alpha \coloneqq \norm{ C_p }_\infty + 3 / 2$, and $\beta \coloneqq 2^\ell \norm{ \chi_0 }_\infty + \norm{ \Gamma }_\infty / 2 + 3 / 2$, where $\chi_0 = [ x_{p,0}^\top \ x_0^\top ]^\top$.
    If $q$ is a prime, and
    \begin{equation}
        \log_2 q > k + \lambda + 2 + \floor*{ \log_2 \max\{ n, p \} \alpha \beta \frac{ c }{ 1 - \gamma} },
        \label{eq:modulus_condition}
    \end{equation}
    then it holds that
    \begin{align*}
        \tilde{x}(t + 1) &= \round*{ \qty( \bar{A} \tilde{x}(t) + \bar{B} \bar{y}(t) ) / 2^\ell } + w(t), \quad \tilde{x}(0) = \bar{x}_0, \\
        \tilde{u}(t) &= \bar{C} \tilde{x}(t) + \bar{D} \bar{y}(t),
    \end{align*}
    for all $t \in \N_0$, where $\tilde{x}(t) = \Reconst \qty(\share{\tilde{x}(t)})$, $\tilde{u}(t) = \Reconst \qty(\share{\tilde{u}(t)})$, $w(t) \in \{-1, 0, 1\}^n$, and $c \ge 1$ and $\gamma \in (0, 1)$ are constants such that $\norm*{ \Phi^t } \le c \gamma^t$ for all $t \in \N_0$.
\end{lemma}

\begin{proof}
    See \appref{app:proof_modulus_condition}.
\end{proof}

\lemref{lem:modulus_condition} implies that the reconstruction of the proposed controller \eqref{eq:2pc_controller} is equivalent to the encoded controller \eqref{eq:encoded_controller} with the error $w(t)$ induced by the $\Trunc$ protocol.
Based on this result, the following theorem states that the control performance of the proposed controller can be arbitrarily close to that of the original controller \eqref{eq:controller} by choosing the bit length of the fractional part $\ell$ to be sufficiently large.

\begin{theorem}
\label{thm:fractional_condition}
    Consider the plants \eqref{eq:plant} and \eqref{eq:plant_2} and controllers \eqref{eq:controller} and \eqref{eq:2pc_controller} under \asmref{asm:fixed-point} and \asmref{asm:stability}.
    Suppose \eqref{eq:modulus_condition} is satisfied.
    For every $\epsilon > 0$, if
    \begin{equation}
        \ell \ge \log_2 \frac{ c \epsilon^{-1} }{ 1 - \gamma } \qty( \frac{ \sqrt{p} }{ 2 } \qty( \norm{ \Gamma } \norm{ \Upsilon } + \norm{ D } ) + 2 \sqrt{n} \norm{ \Upsilon } ),
        \label{eq:fractional_condition}
    \end{equation}
    the controllers guarantee \eqref{eq:error} for all $t \in \N_0$, where $\Upsilon \coloneqq [ D C_p \ C ]$, and $\Gamma$, $c$, and $\gamma$ are defined in \lemref{lem:modulus_condition}.
\end{theorem}

\begin{proof}
    See \appref{app:proof_input_error}.
\end{proof}

Furthermore, if there is no direct feedthrough in the controllers, a more concise condition can be obtained.

\begin{corollary}
\label{cor:fractional_condition}
    Suppose $D = \bfzero_{m \times p}$.
    If
    \begin{equation}
        \ell \ge 2 (k - \ell) + \log_2 \frac{ c \epsilon^{-1} }{ 1 - \gamma } (p + 1) n \sqrt{m},
        \label{eq:fractional_condition_2}
    \end{equation}
    the condition \eqref{eq:fractional_condition} is satisfied.
\end{corollary}

\begin{proof}
    See \appref{app:proof_fractional_condition}.
\end{proof}

Here, $k - \ell$ in \eqref{eq:fractional_condition_2} represents the bit length of the integer part of the controller parameters and initial state.
This result suggests the following procedure for selecting the parameters used in the proposed controller computation protocol:
1) Given the design parameter $\epsilon > 0$ and the security parameter $\lambda \in \N$.
2) Fix $k - \ell \in \N$ to be sufficiently large.
3) Choose the bit length $\ell \in \N$.
4) Design the controller \eqref{eq:controller} to satisfy \asmref{asm:fixed-point} and \asmref{asm:stability}.
5) Check the condition \eqref{eq:fractional_condition} (or \eqref{eq:fractional_condition_2} when $D = \bfzero_{m \times p}$).
If $\ell$ does not satisfy the condition, increase $\ell$ and return to 4).
6) Choose the modulus $q$ to satisfy \eqref{eq:modulus_condition}.

Note that the term $k + \lambda + 2$ in \eqref{eq:modulus_condition} is essential to guarantee the correctness and security of the $\Trunc$ protocol.
The last term of \eqref{eq:modulus_condition} is induced by encoding and mainly depends on $2^\ell$ in $\beta$.
This implies a trade-off between communication costs and precision.
This trade-off can be improved by optimizing the lower bound of $\ell$ in \eqref{eq:fractional_condition}.

\begin{remark}
    During the execution of \proref{pro:controller}, each party $P_i$ receives the $i$th shares in \eqref{eq:parameter_share}, \eqref{eq:output_share}, and \eqref{eq:auxiliary_inputs} from the client and can access public parameters $k$, $\ell$, $q$, and $\lambda$.
    It also obtains the masked messages $(d, e)$ and $m_r$ (if $i = 1$) for each invocation of $\Mult$ and $\Trunc$, respectively.
    Here, the subprotocols are sequentially invoked in the online phase.
    In this case, the modular composition theorem in~\cite{Canetti2000-cl} guarantees that the security of the entire protocol is equivalent to the weakest security of the subprotocols.
    Therefore, the proposed protocol achieves statistical security with respect to $\lambda$.
\end{remark}

\subsection{Comparison to Previous Works}

We compare the theoretical costs in the online phase of the proposed protocol with conventional single-server schemes based on learning with errors (LWE) encryption~\cite{Kim2023-pk} and ring learning with errors (RLWE) encryption~\cite{Teranishi2024-ha,Lee2025-jo}.
\tabref{tab:complexity} lists their computational complexities on the client and server sides, as well as communication complexities from the client (resp. server) to the server (resp. client) and between the servers, which are detailed below.
For a fair comparison, the moduli of the ciphertext spaces and the dimensions of the secret keys in the encryption schemes were set to $q$ and $N$, respectively.

\begin{table*}[t]
    \centering
    \caption{Comparison of computational and communication complexities}
    \label{tab:complexity}
    \begin{tabular}{*{6}c}
        \toprule
        & \multicolumn{2}{c}{Computational complexity} & \multicolumn{3}{c}{Communication complexity} \\
        Scheme & Client & Server & Client $\rightarrow$ Server & Server $\rightarrow$ Client & Server $\leftrightarrow$ Server \\
        \midrule
        \cite{Kim2023-pk} & $O( (m + p) N )$ & $O( (n + m) (n + p) \zeta N^2 )$ & $O( (m + p) N \log q )$ & $O( m N \log q )$ & -- \\
        \cite{Teranishi2024-ha} & $O( N \log N )$ & $O( (n + m + p) \zeta N \log N )$ & $O( N \log q )$ & $O( N \log q )$ & -- \\
        \cite{Lee2025-jo} & $O( m^2 + N \log N )$ & $O( n N \log N )$ & $O( N \log q )$ & $O( N \log q )$ & -- \\
        \proref{pro:controller} & $O( (n + m) (n + p) )$ & $O( (n + m) (n + p) )$ & $O( (n + m) (n + p) \log q )$ & $O( m \log q )$ & $O( (n + m) (n + p) \log q )$ \\
        \bottomrule
    \end{tabular}
    \vspace{-3mm}
\end{table*}

\subsubsection{LWE-based encrypted controller in~\cite{Kim2023-pk}}

The authors of~\cite{Kim2023-pk} represented the controller \eqref{eq:controller} in the form
\begin{align*}
    x'(t + 1) &= T (A - R C) T^{-1} x'(t) + T (B - R D) y(t) + \mu(t), \\
    u(t) &= C T^{-1} x'(t) + D y(t),
\end{align*}
with $R \in \R^{n \times m}$, $T \in \R^{n \times n}$, $x'(t) = T x(t) \in \R^n$, and $\mu(t) = T R u(t) \in \R^m$ such that $T (A - R C) T^{-1} \in \Z^{n \times n}$.
They encrypted each element of its parameters and signals using the Gentry-Sahai-Waters (GSW)~\cite{Gentry2013-eq} and LWE~\cite{Regev2009-ys} encryption schemes, respectively.
The ciphertexts of the GSW and LWE encryptions are a matrix in $\Z_q^{(N + 1) \times \zeta (N + 1)}$ and a vector in $\Z_q^{N + 1}$, where $\zeta = \floor{\log_\nu q}$ for some $\nu \ll q$.

\emph{Communication complexity:}
For each time $t \in \N_0$, the client sends and receives the LWE ciphertexts of $y(t)$ and $u(t)$, respectively, and then returns the re-encryption of $u(t)$ for computing $\mu(t)$.
Thus, the communication complexities from the client to the server and from the server to the client are $O( (m + p) N \log q )$ and $O( m N \log q )$, respectively.

\emph{Computational complexity:}
The computational complexity of the client is $O( (m + p) N )$ because LWE encryption and decryption require an inner product of $N$-dimensional vectors.

On the server, the encrypted controller requires $n (n + 2m + p) + mp$ multiplications between the GSW and LWE ciphertexts and $n (n + 2m + p - 4) + m (p - 1)$ additions of the LWE ciphertexts.
According to~\cite{Lee2025-jo}, each GSW-LWE multiplication comprises $\zeta (N + 1)^2$ scalar multiplications.
In addition, each LWE addition is an $(N + 1)$-dimensional vector addition.
Consequently, the computational complexity of the server is $O( (n + m) (n + p) \zeta N^2 )$.

\subsubsection{RLWE-based encrypted controller in~\cite{Teranishi2024-ha}}

The authors of~\cite{Teranishi2024-ha} used the controller representation $u(t) = H d(t)$ with $d(t) = [ y(t - n)^\top \ \cdots \ y(t)^\top \ u(t - n)^\top \ \cdots \ u(t - 1)^\top ]^\top$ and encrypted $H \in \R^{m \times (n (m + p) + p)}$ and $d(t) \in \R^{n (m + p) + p}$ using the RLWE encryption~\cite{Fan2012-jt}.
The ciphertext of the RLWE encryption is a pair of polynomials in $R_q^2$, where $R_q \coloneqq \Z_q[X] / (X^N + 1)$ is a polynomial ring consisting of $N$-dimensional polynomials in $X$ over $\Z_q$.

\emph{Communication complexity:}
Similar to the LWE-based scheme, for each time $t \in \N_0$, the client sends the encryption of $y(t)$ and re-encryption of $u(t)$ and receives the encryption of $u(t)$.
Meanwhile, the RLWE-based scheme applied message packing~\cite{Smart2012-pc} when encrypting $y(t)$ and $u(t)$, which enables the encryption of a vector as a single ciphertext in $R_q^2$.
Therefore, the communication complexity from the client (resp. server) to the server (resp. client) is $O(N \log q)$.

\emph{Computational complexity:}
The message packing and unpacking are respectively performed before encryption and after decryption.
The packing and unpacking are implemented by the (inverse) number theoretic transform~\cite{Teranishi2024-ha}, of which complexity is $O(N \log N)$~\cite{Satriawan2023-gz}.
The encryption and decryption include a multiplication of polynomials in $R_q$.
The $N$-dimensional polynomial multiplication can also be computed using the number theoretic transform.
Therefore, the computational complexity of the client is $O(N \log N)$.

The encrypted controller requires $n + 1$ multiplications with linearization, $n + m + p + 2$ additions, and $m + p - 1$ rotations of the RLWE ciphertexts.
Note that the multiplication in the RLWE encryption increases the ciphertext size from $R_q^2$ to $R_q^3$.
The linearization reduces it to $R_q^2$, and the rotation replaces the position of vector elements packed in a ciphertext.
Each multiplication and addition process includes four multiplications and two additions in $N$-dimensional polynomials, respectively.
The linearization and rotation carry out $2 \zeta$ $N$-dimensional polynomial multiplications~\cite{Lee2025-jo}, where $\zeta$ is as in the LWE-based scheme.
Therefore, the computational complexity of the server is $O( (n + m + p) \zeta N \log N )$.

\subsubsection{RLWE-based encrypted controller in~\cite{Lee2025-jo}}

When $D = \bfzero_{m \times p}$, the authors of~\cite{Lee2025-jo} employed a variant of the controller representation in~\cite{Teranishi2024-ha}, $u(t) = H' d'(t)$, $d'(t) = [ y(t - 1)^\top \ \cdots \ y(t - n)^\top \ u(t - 1)^\top \ \cdots \ u(t - n)^\top ]^\top$, and encrypted $H' \in \R^{m \times n (m + p)}$ and $d'(t) \in \R^{n (m + p)}$ using the RLWE encryption~\cite{Brakerski2014-ny}.

\emph{Communication complexity:}
The communication complexities are the same as those in~\cite{Teranishi2024-ha}.

\emph{Computational complexity:}
This RLWE-based scheme considers a different encoding of data and replaces the addition, linearization, and rotation on the server in~\cite{Teranishi2024-ha} with aggregation on the client, which consists of $m^2$ scalar additions.
This changes the computational complexities of the client and server to $O(m^2 + N \log N)$ and $O(n N \log N)$, respectively.

\subsubsection{Proposed protocol (\proref{pro:controller})}

First, we note that the computational complexities of the $\Share$ and $\Reconst$ algorithms are $O(1)$ for a scalar message.
Thus, the complexities of (constant) addition/multiplication and truncation of shares are $O(1)$.

\emph{Communication complexity:}
For each time $t \in \N_0$, the client transmits $\share{\tilde{y}(t)}_i$ in \eqref{eq:output_share} and $\aux_i$ in \eqref{eq:auxiliary_inputs} to each party $P_i$, $i \in \{1, 2\}$ and receives $\share{\tilde{u}(t)}_i$ in \eqref{eq:2pc_controller} from the party.
Hence, the communication complexities from the client to the servers and from the servers to the client are $O( (n + m) (n + p) \log q )$ and $O( m \log q )$, respectively.
Furthermore, to compute the controller \eqref{eq:2pc_controller}, the parties invoke the $\Mult$ and $\Trunc$ protocols $(n + m) (n + p)$ and $n$ times, respectively.
Each invocation of $\Mult$ and $\Trunc$ requires the transmission of four and one integers in $\Z_q$ to exchange $(\share{d}, \share{e})$ and $\share{m_r}_2$. 
Therefore, the communication complexity between the servers is $O( (n + m) (n + p) \log q )$.

\emph{Computational complexity:}
The client generates $\share{\tilde{y}(t)}_i$ and $\aux_i$ and reconstructs $\share{\tilde{u}(t)}_i$ in the online phase.
This results in $(n + m) (n + p) + 2n + p$ and $m$ invocations for the $\Share$ and $\Reconst$ algorithms, respectively.
In addition to the $\Mult$ and $\Trunc$ protocols, the computation of \eqref{eq:2pc_controller} requires $n + m$ additions of shares.
Consequently, the computational complexities of the client and servers are $O( (n + m) (n + p) )$.

Compared to the previous schemes, the proposed protocol significantly improves the computational efforts of both the client and servers.
Note that the dimension parameter $N$ is typically chosen to be larger than $2^{10}$~\cite{HES2018}.
On the other hand, although the proposed protocol reduces the communication complexity from the servers to the client, the reduction in the complexity from the client to the servers is limited to satisfying $(n + m) (n + p) \ll N$.
This implies that the proposed protocol incurs high communication costs for large-scale controllers.

\begin{remark}
    One might notice that the computational complexity of the client in \proref{pro:controller} is the same as that of the unencrypted case.
    This is because the client needs to generate the auxiliary inputs $\aux_i$, of which size is $O( (n + m) (n + p) )$.
    However, the auxiliary inputs are independent of the protocol inputs (e.g., the plant output $\hat{y}(t)$).
    This implies that their generation can be addressed in the offline phase if the control system terminates in a finite time.
    Even if this is not the case, it can be outsourced to a third party called a trusted dealer, who is non-colluding with the parties.
    Then, the computational complexity of the client can be reduced to $O(m + p)$.
\end{remark}

\begin{remark}
    \proref{pro:controller} can be extended to $M$ parties using an $M$-out-of-$M$ secret sharing scheme.
    Then, the complexities in \tabref{tab:complexity} increase linearly with $M$.
    Meanwhile, \asmref{asm:adversary} can be relaxed to allow the collusion of any $M - 1$ parties, which mitigates security risks.
    Recall that the proposed scheme is compromised when all parties are corrupted, unlike the single-server setting using homomorphic encryption.
\end{remark}

\section{Improvement of Communication Complexities}
\label{sec:improvement_of_communication_complexities}

As reviewed, the client and parties $P_i$, $i \in \{1, 2\}$ in the proposed protocol exchange $\share{\tilde{y}(t)}_i$ in \eqref{eq:output_share}, $\aux_i$ in \eqref{eq:auxiliary_inputs}, and $\share{\tilde{u}(t)}_i$ in \eqref{eq:2pc_controller} for each time $t \in \N_0$.
Thus, the total communication cost between the client and parties is $2 ( 3 (n + m) (n + p) + 2 n + m + p ) \ell_q$~bit, where $\ell_q = \floor{ \log_2 q } + 1$ is the bit length of the modulus $q$.
This section discusses modifications to the proposed protocol, which improves the cost by applying coordinate transformation and cryptographic primitives.

\subsection{Brunovsky normal form}

It is expected that sparse controller parameters will reduce the number of multiplications of shares by ignoring their zero elements.
To transform the parameters of \eqref{eq:controller} to sparse parameters, we employ the Brunovsky normal form~\cite{Brunovsk-`y1970-kk}.

\begin{definition}[Brunovsky normal form~\cite{Tabuada2003-fp}]
\label{def:bnf}
Consider $F \in \R^{n \times n}$ and $G \in \R^{n \times m}$. For $r \le m$, let $(n_1, \dots, n_r)$ be a sequence of positive integers, such that $n_1 \ge \cdots \ge n_r$ and $n_1 + \cdots + n_r = n$. The pair $(F, G)$ is in the Brunovsky normal form if $F$ and $G$ are of the form
\[
    F =
    \begin{bmatrix}
        F_{n_1} &        & \\
                & \ddots & \\
                &        & F_{n_r}
    \end{bmatrix}, \quad
    G =
    \begin{bmatrix}
        G_{n_1} &        &         & 0      & \cdots & 0 \\
                & \ddots &         & \vdots & \ddots & \vdots \\
                &        & G_{n_r} & 0      & \cdots & 0
    \end{bmatrix},
\]
where $F_{n_i} \in \R^{n_i \times n_i}$ and $G_{n_i} \in \R^{n_i \times 1}$ are of the form
\[
    F_{n_i} =
    \begin{bmatrix}
        0      & 1      & \cdots & 0 \\
        \vdots & \vdots & \ddots & \vdots \\
        0      & 0      & \cdots & 1 \\
        f_{n_i, 0} & f_{n_i, 1} & \cdots & f_{n_i, n_i - 1}
    \end{bmatrix}, \quad
    G_{n_i} =
    \begin{bmatrix}
        0 \\
        \vdots \\
        0 \\
        1
    \end{bmatrix},
\]
for every $i = 1, \dots, r$, and $\sum_{j = 0}^{n_i} f_{n_i, j} X^j$ with $f_{n_i, n_i} = 1$ is the characteristic polynomial of $F_{n_i}$.
\end{definition}

We now consider transforming the dual system of \eqref{eq:controller} into the Brunovsky normal form.
Without loss of generality, suppose that the pair $(A, C)$ of \eqref{eq:controller} is observable.
Moreover, assume that $C$ is full row rank (i.e., $\rank C = m$) for simplicity.
Since $(A^\top, C^\top)$ is controllable, there exists a unique sequence $(n_1, \dots, n_m)$ and nonsingular matrix $P \in \R^{n \times n}$ such that the pair $(P A^\top P^{-1}, P C^\top)$ is in the Brunovsky normal form~\cite{Brunovsk-`y1970-kk,Kalman1972-cq}.
Then, with the coordinate transformation $(P^\top)^{-1} x(t) = [ x_{n_1}(t)^\top \ \cdots \ x_{n_m}(t)^\top ]^\top$, the controller can be decomposed into $m$ multi-input single-output subsystems,
\begin{align*}
    \begin{bmatrix}
        x_{n_1}(t + 1) \\
        \vdots \\
        x_{n_m}(t + 1)
    \end{bmatrix}
    &\!=\! \underbrace{
    \begin{bmatrix}
        A_{n_1} &        & \\
                & \ddots & \\
                &        & A_{n_m}
    \end{bmatrix}
    }_{ (P^\top)^{-1} A P^\top }
    \begin{bmatrix}
        x_{n_1}(t) \\
        \vdots \\
        x_{n_m}(t)
    \end{bmatrix}
    +
    \underbrace{
    \begin{bmatrix}
        B_{n_1} \\
        \vdots \\
        B_{n_m}
    \end{bmatrix}
    }_{ (P^\top)^{-1} B }
    y(t), \\
    \begin{bmatrix}
        u_1(t) \\
        \vdots \\
        u_m(t)
    \end{bmatrix}
    &\!=\! \underbrace{
    \begin{bmatrix}
        C_{n_1} &        & \\
                & \ddots & \\
                &        & C_{n_m}
    \end{bmatrix}
    }_{ C P^\top }
    \begin{bmatrix}
        x_{n_1}(t) \\
        \vdots \\
        x_{n_m}(t)
    \end{bmatrix}
    +
    \underbrace{
    \begin{bmatrix}
        D_1 \\
        \vdots \\
        D_m
    \end{bmatrix}
    }_{ D }
    y(t),
\end{align*}
where $A_{n_i} \in \R^{n_i \times n_i}$ and $C_{n_i} \in \R^{1 \times n_i}$ have the same form as $F_{n_i}^\top$ and $G_{n_i}^\top$ in \defref{def:bnf}, respectively, $B_{n_i} \in \R^{n_i \times p}$, and $D_i \in \R^{1 \times p}$.

We replace $(A, B, C, D)$ of the controller \eqref{eq:controller} with $( (P^\top)^{-1} A P^\top, (P^\top)^{-1} B, C P^\top, D )$ and assume that the parameters satisfy \asmref{asm:fixed-point}.
Using the decomposition, the proposed controller \eqref{eq:2pc_controller} can be implemented as, for every $i = 1, \dots, m$,
\[
    \begin{aligned}
        \share[\big]{\tilde{x}_{n_i}(t + 1)} &\gets \Trunc \qty( \share[\big]{\tilde{A}_{n_i}} \share[\big]{\tilde{x}_{n_i}(t)} + \share[\big]{\tilde{B}_{n_i}} \share[\big]{\tilde{y}(t)}, \ell ), \\
        \share[\big]{\tilde{u}_i(t)} &= \share[\big]{\tilde{x}_{n_i, n_i}(t)} + \share[\big]{\tilde{D}_i} \share[\big]{\tilde{y}(t)},
    \end{aligned}
\]
where 
\[
    \share[\big]{\tilde{A}_{n_i}} \!=\!
    \begin{bmatrix}
        0 & \cdots & 0 & \share[\big]{ \tilde{a}_{n_i, 0} } \\[1pt]
        1 & \cdots & 0 & \share[\big]{ \tilde{a}_{n_i, 1} } \\
        \vdots & \ddots & \vdots & \vdots \\
        0 & \cdots & 1 & \share[\big]{ \tilde{a}_{n_i, n_i - 1} }
    \end{bmatrix}, \ 
    \share[\big]{\tilde{x}_{n_i}(t)} \!=\!
    \begin{bmatrix}
        \share[\big]{\tilde{x}_{n_i, 1}(t)} \\[1pt]
        \share[\big]{\tilde{x}_{n_i, 2}(t)} \\
        \vdots \\
        \share[\big]{\tilde{x}_{n_i, n_i}(t)}
    \end{bmatrix},
\]
and $( a_{n_i, 0}, \dots, \allowbreak a_{n_i, n_i - 1} )$ are the coefficients of the characteristic polynomial of $A_{n_i}$.
The computation of each subsystem requires $n_i + n_i p + p$ multiplications of shares, and hence, $\sum_{i = 1}^m n_i + n_i p + p = (n + m) p + n$ Beaver triples are consumed to evaluate the $m$ subsystems.
The total communication cost is then given as $2 ( 3 (n + m) p + 5 n + m + p ) \ell_q$~bit, thereby reducing the communication complexity from the client to the servers to $O( (n + m) p \log q )$.
The modification also improves the computational complexities of the client and servers to $O( (n + m) p )$.

\subsection{Pseudorandom function}

By definition, the shares for the party $P_1$ are random numbers in $\Z_q$ and independent of messages.
This characteristic allows us to utilize a pseudorandom function further to reduce the communication cost between the client and parties.
A pseudorandom function is a polynomial-time deterministic algorithm $\PRF : \{0, 1\}^{\ell_\mathrm{key}(\lambda)} \times \{0, 1\}^{\ell_\mathrm{in}(\lambda)} \to \{0, 1\}^{\ell_\mathrm{out}(\lambda)}$, which takes an $\ell_\mathrm{key}(\lambda)$-bit key and an $\ell_\mathrm{in}(\lambda)$-bit string as input and outputs an $\ell_\mathrm{out}(\lambda)$-bit string, where $\lambda$ is a security parameter~\cite{Katz2014-kb}.
It is necessary that any polynomial-time algorithm without knowledge of the key cannot distinguish the output from a uniformly random string over $\{0, 1\}^{\ell_\mathrm{out}(\lambda)}$.

Let $\PRF_K$ be a pseudorandom function with a key $K \in \{0, 1\}^{\ell_\mathrm{key}(\lambda)}$ whose input and output domains are $\{0, 1\}^{\ell_\mathrm{in}(\lambda)}$ and $\Z_q$, respectively.
To generate shares for $P_1$, in the offline phase, the client randomly chooses a key $K$ and transmits $\PRF_K$ to $P_1$.
In the online phase, the client and $P_1$ can generate the same pseudorandom number $r \gets \PRF_K(t \bmod 2^{\ell_\mathrm{in}(\lambda)})$ regardless of $m \in \Z_q$.
Then, $P_1$ sets $\share{m}_1 = r$, and the client computes the share for $P_2$ as $\share{m}_2 = m - r \bmod q$, which satisfies the requirement \refeq{eq:correctness}.
By applying this technique to each share generation, all communications from the client to $P_1$ can be eliminated.
This reduces the total communication cost between the client and parties to $( 3 (n + m) p + 5 n + 2 m + p ) \ell_q$~bit.
Note that, because $\PRF_K$ is deterministic, the client must refresh the key $K$ every $\tau \in \N$ time step to ensure the security of $\PRF_K$.
Although this process incurs additional communication costs, the modified share generation remains more efficient than the original one.

\subsection{Pseudorandom correlation function}

Instead of applying the pseudorandom functions, one can use a pseudorandom correlation function~\cite{Boyle2020-sh}, which is analogous to a pseudorandom function that generates correlated randomness.
A pseudorandom correlation function $\PCF$ is a two-party protocol in which each party $P_i$, $i \in \{1, 2\}$ is given a key $K_i$ and, for some input string $x$, the parties can locally compute correlated strings $(\PCF_{K_1}(x), \PCF_{K_2}(x))$, where the outputs and random samples from the target correlation are indistinguishable~\cite{Couteau2023-fs}.

In our scenario, the keys can be generated by the client, allowing the parties to obtain Beaver triples $(\share{a}_i, \share{b}_i, \share{c}_i)$ on demand during the online phase by invoking a pseudorandom correlation function $\PCF_{K_i}$ for the Beaver triple correlation with a common reference string $\mathsf{crs}$ (e.g., time $t$), namely $(\share{a}_i, \share{b}_i, \share{c}_i) \gets \PCF_{K_i}(\mathsf{crs})$.
Moreover, the shares of random numbers $(\share{r_h}, \share{r'_h})$ can be generated in the same manner.
This approach eliminates the transmission of the auxiliary inputs $\aux_i$, thereby reducing the total communication cost between the client and parties to $2 (m + p) \ell_q$~bit.
It also decreases the computational complexity of the client to $O(m + p)$, as the client only needs to execute the $\Share$ and $\Reconst$ algorithms $p$ and $m$ times, respectively.

Integrating the proposed protocol with pseudorandom correlation functions is the most promising approach to significantly improve the communication complexity.
In this pursuit, exploring efficient constructions of such pseudorandom correlation functions in real-time control systems is future work.

\section{Numerical Examples}
\label{sec:numerical_examples}

Consider a PID controller of parallel form in~\cite{Kim2023-pk}.
The controller parameters of \eqref{eq:controller} are $A = \qty[ \begin{smallmatrix} 2 - N_d & N_d - 1 \\ 1 & 0 \end{smallmatrix} ]$, $B = [1 \ 0]^\top$, $C = [c_1 \ c_2]$, and $D = d$, where $c_1 = K_i T_s - N_d^2 K_d / T_s$, $c_2 = N_d K_i T_s - K_i T_s + N_d^2 K_d / T_s$, and $d = K_p + N_d K_d / T_s$.
Here, $K_p$, $K_i$, and $K_d$ are proportional, integral, and derivative gains, respectively, $N_d \in \N$ is the parameter of a derivative filter, and $T_s \in \R$ is a sampling time.
Note that, if $2 - N_d, N_d - 1 \in \ZZ{k - \ell}$, the parameters $A$ and $B$ are already integers.
Then, the proposed controller \eqref{eq:2pc_controller} can be simplified as
\begin{align*}
    \share[\big]{\tilde{x}(t + 1)} &= \share[\big]{A} \share[\big]{\tilde{x}(t)} + \share[\big]{B} \share[\big]{\tilde{y}(t)}, \\
    \share[\big]{\tilde{u}(t)} &= \share[\big]{\tilde{C}} \share[\big]{\tilde{x}(t)} + \share[\big]{\tilde{D}} \share[\big]{\tilde{y}(t)},
\end{align*}
and the $\Trunc$ protocol is not necessary to implement \eqref{eq:controller}.

Suppose that the plant \eqref{eq:plant} and sampling time are given by the single-input and single-output system in ~\cite[Eq.~(2), with $\alpha = 0.2$]{Astrom2000-ta} and $T_s = 0.1$~s, respectively.
Let $\epsilon = 2^{-10}$, $\lambda = 80$~bit, $k - \ell = 8$~bit, $x_{p,0} = [100 \ 100 \ 100 \ 100]^\top$, and $x_0 = [0 \ 0]^\top$.
For instance, the PID controller with the parameters $A = \qty[ \begin{smallmatrix} 1 & 0 \\ 1 & 0 \end{smallmatrix} ]$, $B = [1 \ 0]^\top$, $C = [2.7368927 \ {-2.96540833}]$, and $D = -5.01071167$ satisfies \asmref{asm:fixed-point}, \asmref{asm:stability}, and \eqref{eq:fractional_condition} with $\ell = 32, 40, 48, 56$~bit.
The bit length of modulus $q$ can be chosen as $256$~bit to satisfy \eqref{eq:modulus_condition} for all the choices of $\ell$.
\figref{fig:pid_error} shows the input errors $\norm{ u(t) - \hat{u}(t) }$.
In this figure, the input errors remain below $\epsilon = 2^{-10}$ for all $\ell = 32, 40, 58, 64$ and $t = 0 , \dots, 50$, thereby supporting \thmref{thm:fractional_condition}.
Furthermore, the figure suggests that increasing the bit length $\ell$ can reduce the errors.

\begin{figure}[t]
    \centering
    \includegraphics[scale=1,trim=0 6 0 0,clip]{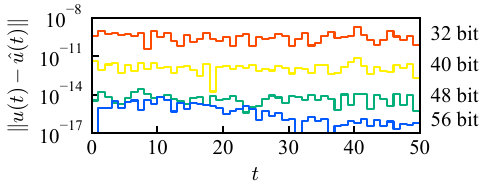}
    \caption{Input errors between the original and two-party PID controls with $\ell = 32, 40, 48, 56$~bit.}
    \label{fig:pid_error}
\end{figure}

Next, we consider an observer-based controller for the plant \eqref{eq:plant} of multi-input and multi-output.
The controller parameters of \eqref{eq:controller} are given as $A = A_p + B_p F - L C_p$, $B = L$, $C = F$, $D = \bfzero_{m \times p}$, where $F$ and $L$ are feedback gains.
Because the controller parameters are not necessarily integers, we used the proposed controller \eqref{eq:2pc_controller} with the $\Trunc$ protocol.

Suppose the plant is given as the four-tank process in~\cite{Johansson2000-il} with $h_1^0=12.4$~cm, $h_2^0=12.7$~cm, $h_3^0=1.8$~cm, $h_4^0=1.4$~cm, $v_1^0=3$~V, $v_2^0=3$~V, $\gamma_1=0.7$, $\gamma_2=0.6$ and $T_s = 0.5$~s.
Let $\epsilon$, $\lambda$, $k - \ell$, and $\ell$ be as in the PID control example.
Let $x_{p,0} = [10 \ 10 \ 10 \ 10]^\top$ and $x_0 = [0 \ 0 \ 0 \ 0]^\top$.
The controller with the parameters,
\begin{align*}
    A &=
    \begin{bmatrix*}[r]
        \scriptstyle  0.56817627 & \scriptstyle -0.00167847 & \scriptstyle  0.01213074 & \scriptstyle -0.00909424 \\[-3pt]
        \scriptstyle -0.00201416 & \scriptstyle  0.57826233 & \scriptstyle -0.00939941 & \scriptstyle  0.00976562 \\[-3pt]
        \scriptstyle -0.15261841 & \scriptstyle -0.01811218 & \scriptstyle  0.97219849 & \scriptstyle -0.00508118 \\[-3pt]
        \scriptstyle -0.01197815 & \scriptstyle -0.15417480 & \scriptstyle -0.00314331 & \scriptstyle  0.98011780
    \end{bmatrix*}, \\
    B &=
    \begin{bmatrix*}[c]
        \scriptstyle 0.78367615 & \scriptstyle 0          & \scriptstyle 0.30230713 & \scriptstyle 0          \\[-3pt]
        \scriptstyle 0          & \scriptstyle 0.78463745 & \scriptstyle 0          & \scriptstyle 0.30718994
    \end{bmatrix*}^\top, \\
    C &=
    \begin{bmatrix*}[r]
        \scriptstyle -0.77249146 & \scriptstyle -0.03674316 & \scriptstyle -0.20259094 & \scriptstyle -0.21775818 \\[-3pt]
        \scriptstyle -0.06149292 & \scriptstyle -0.76373291 & \scriptstyle -0.29937744 & \scriptstyle -0.21432495
    \end{bmatrix*},
\end{align*}
and a $256$-bit prime $q$ satisfy \asmref{asm:fixed-point}, \asmref{asm:stability}, \eqref{eq:modulus_condition}, and \eqref{eq:fractional_condition}.
Similar to the PID control example, \figref{fig:observer_error} shows the input errors for $\ell = 32, 40, 48, 56$~bit.
This figure demonstrates that the performance degradation in the proposed protocol is negligible even when using the truncation protocol.

\begin{figure}[t]
    \centering
    \includegraphics[scale=1,trim=0 6 0 0,clip]{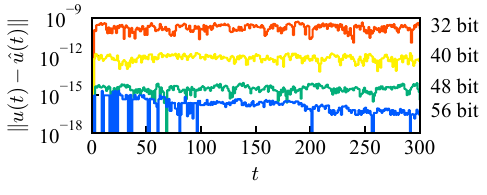}
    \caption{Input errors between the original and two-party observer-based controls with $\ell = 32, 40, 48, 56$~bit.}
    \label{fig:observer_error}
    \vspace{-5mm}
\end{figure}

\section{Conclusion}
\label{sec:conclusion}

We proposed a two-party protocol for the computation of a dynamic controller.
The proposed protocol was realized based on a secret sharing scheme without the need for controller state decryption or input re-encryption and achieved statistical security.
The performance of the dynamic controller under the proposed protocol can be made arbitrarily close to that of the original controller by approximating the controller parameters and initial state as sufficiently precise fixed-point numbers.
We compared the computational and communication complexities of the protocol with those of conventional encrypted controllers and presented some improvements to the protocol using coordinate transformation and cryptographic primitives.
Furthermore, the feasibility of the protocol was demonstrated through numerical examples.

In future work, we will extend the protocol to more advanced control strategies, such as model predictive control.
We will also consider guaranteeing security not only against semi-honest adversaries but also against active adversaries who deviate from our intended behavior.
This would address the limitation imposed by \asmref{asm:adversary}.

\appendix

\subsection{Proof of \lemref{lem:truncation}}
\label{app:proof_truncation}

By the conditions of the bit lengths of $m$, $r$, and $r'$, it follows from \refeq{eq:constant_addition} to \refeq{eq:addition} that $\share{m_r} = \share{ m + 2^\ell r + r' + 2^{\ell - 1} \bmod q } = \share{ m + 2^\ell r + r' + 2^{\ell - 1} }$.
Thus, from \refeq{eq:correctness}, we have $m_r = \Reconst(\share{m_r}) = m + 2^\ell r + r' + 2^{\ell - 1}$.
It then holds that $m_r - 2^{\ell - 1} \bmod 2^\ell = m + r' \bmod 2^\ell = (m \bmod 2^\ell) + (r' \bmod 2^\ell) + 2^\ell w = (m \bmod 2^\ell) + r' + 2^\ell w$, where $r' \in \ZZ{\ell}$ and $w \in \{-1, 0, 1\}$.
The properties from \refeq{eq:constant_addition} to \refeq{eq:addition} yield
\begin{align*}
    \share{m'}
    &= \share{ \inv(2^\ell, q) ( m + r' - ( m_r - 2^{\ell - 1} \bmod 2^\ell ) ) \bmod q }, \\
    &= \share{ \inv(2^\ell, q) ( m - ( m \bmod 2^\ell ) + 2^\ell w ) \bmod q }, \\
    &= \share*{ \round*{ \frac{ m }{ 2^\ell } } + w \bmod q } = \share*{ \round*{ \frac{ m }{ 2^\ell } } + w },
\end{align*}
where $\abs*{ \round{ m / 2^\ell } + w } \le 2^{\kappa - \ell - 1} + 1 < q / 2$.
Therefore, the claim holds from \refeq{eq:correctness}.

\subsection{Proof of \lemref{lem:modulus_condition}}
\label{app:proof_modulus_condition}

We first note that $c$ and $\gamma$ in the claim exist from \asmref{asm:stability}.
In addition, $\alpha > 1$, $\beta > 1$, and $c / ( 1 - \gamma ) > 1$ are fulfilled.
Let $\kappa = \floor{ \log_2 q } - \lambda - 1$, and then, $\ZZ{k} \subset \ZZ{\kappa} \subset \Z_q$ holds because $\log_2 q > \kappa + \lambda + 1$ and $\kappa \ge k + 1 + \floor{ \log_2 \max\{ n, p \} \alpha \beta \frac{ c }{ 1 - \gamma} }$.
This implies that $\tilde{A} = \bar{A}$, $\tilde{B} = \bar{B}$, $\tilde{C} = \bar{C}$, $\tilde{D} = \bar{D}$, and $\tilde{x}_0 = \bar{x}_0$.
From \refeq{eq:correctness}, \refeq{eq:addition}, and \refeq{eq:multiplication}, if $\tilde{A} \tilde{x}(t) + \tilde{B} \tilde{y}(t) \in \ZZ{\kappa}^n$ and $\bar{y}(t) \in \Z_q^p$, \lemref{lem:truncation} guarantees that
\begin{align*}
    \tilde{x}(t + 1)
    &= \Reconst \qty( \Trunc \qty( \share[\big]{\tilde{A}} \share[\big]{\tilde{x}(t)} + \share[\big]{\tilde{B}} \share[\big]{\tilde{y}(t)}, \ell ) ), \\
    &= \round*{ \qty( \bar{A} \tilde{x}(t) + \bar{B} \bar{y}(t) ) / 2^\ell } + w(t).
\end{align*}
Similary, if $\tilde{C} \tilde{x}(t) + \tilde{D} \tilde{y}(t) \in \Z_q^m$ and $\bar{y}(t) \in \Z_q^p$, it follows that
\begin{align*}
    \tilde{u}(t)
    &= \Reconst \qty( \share[\big]{\tilde{C}} \share[\big]{\tilde{x}(t)} + \share[\big]{\tilde{C}} \share[\big]{\tilde{y}(t)} ), \\
    &= \bar{C} \tilde{x}(t) + \bar{D} \bar{y}(t).
\end{align*}
Hence, the claim holds if $\tilde{A} \tilde{x}(t) + \tilde{B} \tilde{y}(t) \in \ZZ{\kappa}^n$, $\tilde{C} \tilde{x}(t) + \tilde{D} \tilde{y}(t) \in \Z_q^m$, and $\bar{y}(t) \in \Z_q^p$ for all $t \in \N_0$.
The remainder of this proof shows this through induction.

Let $\hat{x}(t) = 2^{-\ell} \tilde{x}(t)$ and $\hat{\chi}(t) = [ \hat{x}_p(t)^\top \ \hat{x}(t)^\top ]^\top$.
For $t = 0$, it holds that
\[
    \norm\big{ \bar{y}(0) }_\infty
    \le \norm\big{ 2^\ell \hat{y}(0) }_\infty + \frac{ 1 }{ 2 }
    \le 2^\ell \norm\big{ C_p }_\infty \norm\big{ \hat{\chi}(0) }_\infty + \frac{ 1 }{ 2 }
    < \alpha \beta,
\]
and
\begin{align*}
    &\max\qty{ \norm\big{ \tilde{A} \tilde{x}(0) + \tilde{B} \tilde{y}(0) }_\infty, \norm\big{ \tilde{C} \tilde{x}(0) + \tilde{D} \tilde{y}(0) }_\infty } \\
    &\le 2^{k - 1} \max\{ n, p \} \qty( \norm*{ \bar{x}_0 }_\infty + \norm*{ \bar{y}(0) }_\infty ), \\
    &\le 2^{k - 1} \max\{ n, p \} \qty( 2^\ell \qty( 1 + \norm*{ C_p }_\infty ) \norm*{ \hat{\chi}(0) }_\infty + \frac{ 1 }{ 2 }), \\
    &< 2^{k - 1} \max\{ n, p \} \alpha \beta,
\end{align*}
which means $\tilde{A} \tilde{x}(0) + \tilde{B} \tilde{y}(0) \in \ZZ{\kappa}^n$, $\tilde{C} \tilde{x}(0) + \tilde{D} \tilde{y}(0) \in \ZZ{\kappa}^m \subset \Z_q^m$, and $\tilde{y}(0) = \bar{y}(0) \in \ZZ{\kappa}^p \subset \Z_q^p$ because $\log_2 \alpha \beta < k - 1 + \log_2 \max\{ n, p \} \alpha \beta < \kappa$.

Assume that $\tilde{A} \tilde{x}(t) + \tilde{B} \tilde{y}(t) \in \ZZ{\kappa}^n$, $\tilde{C} \tilde{x}(t) + \tilde{D} \tilde{y}(t) \in \Z_q^m$, and $\tilde{y}(t) = \bar{y}(t) \in \Z_q^p$ hold for all $t = 0, \dots, \tau - 1$ and $\tau \in \N$.
With \refeq{eq:correctness}, \refeq{eq:addition}, and \refeq{eq:multiplication}, it follows from \eqref{eq:plant_2} to \eqref{eq:approximate_input} that, for all $t = 0, \dots, \tau - 1$,
\begin{align*}
    \hat{x}_p(t + 1)
    &= A_p \hat{x}_p(t) + 2^{-2 \ell} B_p \tilde{u}(t), \\
    &= ( A_p + B_p D C_p ) \hat{x}_p(t) + B_p C \hat{x}(t) + B_p D d_1(t), \\
    \hat{x}(t + 1)
    &= 2^{-\ell} \tilde{x}(t + 1), \\
    &= B C_p \hat{x}_p(t) + A \hat{x}(t) + B d_1(t) + d_2(t),
\end{align*}
where $d_1(t) = 2^{-\ell} \bar{y}(t) - \hat{y}(t)$, $d_2(t) = 2^{-\ell} ( e(t) + w(t) )$, and $e(t) = \round{ A \tilde{x}(t) + B \bar{y}(t) } - ( A \tilde{x}(t) + B \bar{y}(t) )$.
Hence, we have
\begin{equation}
    \hat{\chi}(t + 1) = \Phi \hat{\chi}(t) + \Gamma d_1(t) + \Xi d_2(t)
    \label{eq:2pc_closed-loop}
\end{equation}
for all $t = 0, \dots, \tau - 1$, where $\Xi = [ \bfzero_{n_p \times n}^\top \ I_n ]^\top$.
This implies $\hat{\chi}(\tau) = \Phi^\tau \hat{\chi}(0) + \sum_{s = 0}^{\tau - 1} \Phi^{\tau - 1- s} \qty( \Gamma d_1(s) + \Xi d_2(s) )$.

Using this equation, the norm of $\tilde{x}(\tau)$ is bounded by
\begin{align*}
    & \norm*{ \tilde{x}(\tau) }_\infty = 2^\ell \norm*{ \hat{x}(\tau) }_\infty \le 2^\ell \norm*{ \hat{\chi}(\tau) }_\infty, \\
    &\le 2^\ell \qty( \norm*{ \hat{\chi}(0) }_\infty + \frac{ 2^{-\ell} }{ 2 } \norm*{ \Gamma }_\infty + \frac{ 2^{-\ell} \cdot 3 }{ 2 } \norm*{ \Xi }_\infty ) \sum_{s = 0}^\tau \norm*{ \Phi^s }_\infty, \\
    &< \qty( 2^\ell \norm*{ \hat{\chi}(0) }_\infty + \frac{ 1 }{ 2 } \norm*{ \Gamma }_\infty + \frac{ 3 }{ 2 } ) \frac{ c }{ 1 - \gamma } = \beta \frac{ c }{ 1 - \gamma },
\end{align*}
where $\norm{ d_1(s) }_\infty \le 2^{-\ell} \cdot ( 1 / 2 )$, $\norm{ d_2(s) }_\infty \le 2^{-\ell} \cdot ( 3 / 2 )$, and $\sum_{s = 0}^k \norm{ \Phi^s }_\infty \le \sum_{s = 0}^k \norm{ \Phi^s } \le \sum_{s = 0}^k c \gamma^s < \sum_{s = 0}^\infty c \gamma^s = c / (1 - \gamma)$.
Moreover, the norm of $\bar{y}(\tau)$ is bounded by
\begin{align*}
    \norm*{ \bar{y}(\tau) }_\infty
    &\le 2^\ell \norm*{ C_p }_\infty \norm*{ \hat{\chi}(\tau) }_\infty + \frac{ 1 }{ 2 } \le \norm*{ C_p }_\infty \beta \frac{ c }{ 1 - \gamma } + \frac{ 1 }{ 2 }, \\
    &< \qty( \norm*{ C_p }_\infty + \frac{ 1 }{ 2 } ) \beta \frac{ c }{ 1 - \gamma } < \alpha \beta \frac{ c }{ 1 - \gamma },
\end{align*}
which means $\tilde{y}(\tau) = \bar{y}(\tau) \in \ZZ{\kappa}^p \subset \Z_q^p$ because
\begin{align*}
    \alpha \beta \frac{ c }{ 1 - \gamma }
    &<2^{k - 1} \max\{ n, p \} \alpha \beta \frac{ c }{ 1 - \gamma }, \\
    &< 2^{k - 1} \qty( \floor*{ \max\{ n, p \} \alpha \beta \frac{ c }{ 1 - \gamma } } + 1 ), \\
    &\le 2^{k - 1 + \floor{ \max\{ n, p \} \alpha \beta c / ( 1 - \gamma ) }} = 2^{\kappa - 2} < 2^{\kappa - 1} - 1.
\end{align*}
From these bounds, it holds that
\begin{align*}
    &\max\qty{ \norm\big{ \tilde{A} \tilde{x}(\tau) + \tilde{B} \tilde{y}(\tau) }_\infty, \norm\big{ \tilde{C} \tilde{x}(\tau) + \tilde{D} \tilde{y}(\tau) }_\infty } \\
    &\le 2^{k - 1} \max\{ n, p \} \qty( \norm{ \tilde{x}(\tau) }_\infty + \norm{ \bar{y}(\tau) }_\infty ), \\
    &< 2^{k - 1} \max\{ n, p \} \qty( 1 + \norm*{ C_p }_\infty + \frac{ 1 }{ 2 } ) \beta \frac{ c }{ 1 - \gamma }, \\
    &= 2^{k - 1} \max\{ n, p \} \alpha \beta \frac{ c }{ 1 - \gamma } < 2^{\kappa - 1} - 1.
\end{align*}
Therefore, $\tilde{A} \tilde{x}(\tau) + \tilde{B} \tilde{y}(\tau) \in \ZZ{\kappa}^n$ and $\tilde{C} \tilde{x}(\tau) + \tilde{D} \tilde{y}(\tau) \in \ZZ{\kappa}^m \subset \Z_q^m$ are satisfied.

\subsection{Proof of \thmref{thm:fractional_condition}}
\label{app:proof_input_error}

Let $\hat{x}(t) = 2^\ell \tilde{x}(t)$, $\chi(t) = [ x_p(t)^\top \ x(t)^\top ]^\top$, and $\hat{\chi}(t) = [ \hat{x}_p(t)^\top \ \hat{x}(t)^\top ]^\top$.
Define $\varepsilon(t) = \chi(t) - \hat{\chi}(t)$.
If \eqref{eq:modulus_condition} is satisfied, from \eqref{eq:closed-loop} and \eqref{eq:2pc_closed-loop}, the dynamics of $\varepsilon$ is given as
\[
    \varepsilon(t + 1) = \Phi \varepsilon(t) - \Gamma d_1(t) - \Xi d_2(t), \quad \varepsilon(0) = \bfzero_{(n_p + n) \times 1},
\]
where $\Xi$, $d_1(t)$, and $d_2(t)$ are the same as in the proof of \lemref{lem:modulus_condition}.
From \asmref{asm:stability} and the definitions of $d_1(t) \in \R^p$ and $d_2(t) \in \R^n$, it holds that, for all $t \in \N_0$, $\norm{ \Phi^t } \le c \gamma^t$, $\norm{ d_1(t) } \le 2^{-\ell - 1} \sqrt{p}$, and $\norm{ d_2(t) } < 2^{-\ell + 1} \sqrt{n}$.
Applying these bounds to the norm of $\varepsilon(t)$, we obtain
\begin{align*}
    \norm{ \varepsilon(t) }
    &= \norm{ \Phi^t \varepsilon(0) + \sum_{s = 0}^{t - 1} \Phi^{t - 1 - s} \qty( \Gamma d_1(s) + \Xi d_2(s) ) }, \\
    &\le 2^{-\ell} \qty( \frac{ \sqrt{p} }{ 2 } \norm*{ \Gamma } + 2 \sqrt{n} \norm*{ \Xi } ) \sum_{s = 0}^{t - 1} \norm*{ \Phi^s }, \\
    &< 2^{-\ell} \qty( \frac{ \sqrt{p} }{ 2 } \norm{ \Gamma } + 2 \sqrt{n} ) \frac{ c }{ 1 - \gamma },
\end{align*}
where $\sum_{s = 0}^{t - 1} \norm*{ \Phi^s } < c / ( 1 - \gamma )$.

From \lemref{lem:modulus_condition}, \eqref{eq:controller}, and \eqref{eq:approximate_input}, the error of the control inputs is given as
\begin{align*}
    u(t) - \hat{u}(t)
    &= ( C x(t) + D y(t) ) - 2^{-2 \ell} \qty( \bar{C} \tilde{x}(t) + \bar{D} \bar{y}(t) ), \\
    &= D C_p e_p(t) + C e(t) + D d_1(t), \\
    &= \Upsilon \varepsilon(t) + D d_1(t).
\end{align*}
With the bounds of $\varepsilon(t)$ and $d_1(t)$, it concludes that
\begin{align*}
    &\norm*{ u(t) - \hat{u}(t) }
    \le \norm{ \Upsilon } \norm{ \varepsilon(t) } + \norm{ D } \norm{ d_1(t) }, \\
    &\le 2^{-\ell} \qty( \frac{ \sqrt{p} }{ 2 } \norm{ \Gamma } \norm{ \Upsilon } + 2 \sqrt{n} \norm{ \Upsilon } ) \frac{ c }{ 1 - \gamma } + 2^{-\ell - 1} \sqrt{p} \norm{ D }, \\
    &< 2^{-\ell} \qty[ \frac{ \sqrt{p} }{ 2 } \qty( \norm{ \Gamma } \norm{ \Upsilon } + \norm{ D } ) + 2 \sqrt{n} \norm{ \Upsilon } ] \frac{ c }{ 1 - \gamma }.
\end{align*}
Therefore, \eqref{eq:error} holds if $\ell$ satisfies \eqref{eq:fractional_condition}.

\subsection{Proof of \corref{cor:fractional_condition}}
\label{app:proof_fractional_condition}

The claim can be obtained by applying the bounds $\norm{ \Gamma } = \norm{ B } \le 2^{k - \ell - 1} \sqrt{ np }$, $\norm{ \Upsilon } = \norm{ C } \le 2^{k - \ell - 1} \sqrt{ mn }$, and $\norm{ D } = 0$ to \eqref{eq:fractional_condition}.

\bibliographystyle{IEEEtran}
\bibliography{reference}

@INCOLLECTION{Escudero2020-ih,
title = "Improved primitives for {MPC} over mixed arithmetic-binary circuits",
author = "Escudero, Daniel and Ghosh, Satrajit and Keller, Marcel and Rachuri, Rahul and Scholl, Peter",
booktitle = "Advances in Cryptology -- CRYPTO 2020",
publisher = "Springer International Publishing",
pages = "823--852",
year =  2020
}

@ARTICLE{Boyle2020-sh,
title = "Correlated pseudorandom functions from variable-density {LPN}",
author = "Boyle, Elette and Couteau, Geoffroy and Gilboa, N and Ishai, Yuval and Kohl, Lisa and Scholl, Peter",
journal = "IEEE Annual Symposium on Foundations of Computer Science",
pages = "1069--1080",
year =  2020
}

@ARTICLE{Kim2023-pk,
title = "Dynamic Controller That Operates Over Homomorphically Encrypted Data for Infinite Time Horizon",
author = "Kim, Junsoo and Shim, Hyungbo and Han, Kyoohyung",
journal = "IEEE Transactions on Automatic Control",
publisher = "IEEE",
volume =  68,
number =  2,
pages = "660--672",
year =  2023
}

@ARTICLE{Teranishi2024-ha,
title = "Input-output history feedback controller for encrypted control with leveled fully homomorphic encryption",
author = "Teranishi, Kaoru and Sadamoto, Tomonori and Kogiso, Kiminao",
journal = "IEEE Transactions on Control of Network Systems",
publisher = "Institute of Electrical and Electronics Engineers (IEEE)",
volume =  11,
number =  1,
pages = "271--283",
year =  2024
}

@ARTICLE{Lee2025-jo,
title = "Encrypted dynamic control exploiting limited number of multiplications and a method using {RLWE}-based cryptosystem",
author = "Lee, Joowon and Lee, Donggil and Kim, Junsoo and Shim, Hyungbo",
journal = "IEEE Transactions on Systems, Man, and Cybernetics: Systems",
publisher = "Institute of Electrical and Electronics Engineers (IEEE)",
volume =  55,
number =  1,
pages = "158--169",
year =  2025
}

@ARTICLE{Astrom2000-ta,
title = "Benchmark systems for {PID} control",
author = "\AA{}str{\"{o}}m, K J and H{\"{a}}gglund, T",
journal = "IFAC Proceedings Volumes",
publisher = "Elsevier BV",
volume =  33,
number =  4,
pages = "165--166",
year =  2000
}

@ARTICLE{Johansson2000-il,
title = "The quadruple-tank process: A multivariable laboratory process with an adjustable zero",
author = "Johansson, K H",
journal = "IEEE Transactions on Control Systems Technology",
publisher = "Institute of Electrical and Electronics Engineers (IEEE)",
volume =  8,
number =  3,
pages = "456--465",
year =  2000
}

@INCOLLECTION{Couteau2023-fs,
title = "Pseudorandom correlation functions from variable-density {LPN}, revisited",
author = "Couteau, Geoffroy and Ducros, Cl\'{e}ment",
booktitle = "Lecture Notes in Computer Science",
publisher = "Springer Nature Switzerland",
pages = "221--250",
year =  2023
}

@BOOK{Katz2014-kb,
title = "Introduction to modern cryptography",
author = "Katz, Jonathan and Lindell, Yehuda",
publisher = "Chapman \& Hall/CRC",
edition =  2,
year =  2014
}

@ARTICLE{Smart2012-pc,
title = "Fully homomorphic {SIMD} operations",
author = "Smart, N P and Vercauteren, F",
journal = "Designs Codes and Cryptography",
volume =  71,
pages = "57--81",
year =  2012
}

@INCOLLECTION{Tabuada2003-fp,
title = "Model checking {LTL} over controllable linear systems is decidable",
author = "Tabuada, Paulo and Pappas, George J",
booktitle = "Hybrid Systems: Computation and Control",
publisher = "Springer Berlin Heidelberg",
pages = "498--513",
year =  2003
}

@ARTICLE{Brunovsk-`y1970-kk,
title = "A classification of linear controllable systems",
author = "Brunovsk\`y, Pavol",
journal = "Kybernetika",
publisher = "Institute of Information Theory and Automation AS CR",
volume =  6,
number =  3,
pages = "173--188",
year =  1970
}

@INCOLLECTION{Kalman1972-cq,
title = "Kronecker invariants and feedback",
author = "Kalman, R E",
booktitle = "Ordinary Differential Equations",
publisher = "Elsevier",
pages = "459--471",
year =  1972
}

@techreport{HES2018,
 author = {Martin Albrecht and Melissa Chase and Hao Chen and Jintai Ding and Shafi Goldwasser and Sergey Gorbunov and Shai Halevi and Jeffrey Hoffstein and Kim Laine and Kristin Lauter and Satya Lokam and Daniele Micciancio and Dustin Moody and Travis Morrison and Amit Sahai and Vinod Vaikuntanathan},
 title = {Homomorphic Encryption Security Standard},
 institution= {HomomorphicEncryption.org},
 publisher = {HomomorphicEncryption.org},
 address = {Toronto, Canada},
 year = {2018},
 month = {November}
 }

@BOOK{Cramer2015-sj,
title = "Secure multiparty computation and secret sharing",
author = "Cramer, Ronald and Damgaard, Ivan Bjerre and Nielsen, Jesper Buus",
publisher = "Cambridge University Press",
year =  2015
}

@INCOLLECTION{Beaver2007-kr,
title = "Efficient multiparty protocols using circuit randomization",
author = "Beaver, Donald",
booktitle = "Advances in Cryptology -- CRYPTO '91",
publisher = "Springer Berlin Heidelberg",
pages = "420--432",
year =  1992
}

@ARTICLE{Canetti2000-cl,
title = "Security and composition of multiparty cryptographic protocols",
author = "Canetti, Ran",
journal = "Journal of Cryptology",
publisher = "Springer Science and Business Media LLC",
volume =  13,
number =  1,
pages = "143--202",
year =  2000
}

@ARTICLE{Shannon1949-mc,
title = "Communication theory of secrecy systems",
author = "Shannon, C E",
journal = "The Bell System technical journal",
publisher = "Institute of Electrical and Electronics Engineers (IEEE)",
volume =  28,
number =  4,
pages = "656--715",
year =  1949
}

@ARTICLE{Regev2009-ys,
title = "On lattices, learning with errors, random linear codes, and cryptography",
author = "Regev, Oded",
journal = "Journal of the ACM",
publisher = "Association for Computing Machinery (ACM)",
volume =  56,
number =  6,
pages = "1--40",
year =  2009
}

@ARTICLE{Brakerski2014-ny,
title = "(Leveled) fully homomorphic encryption without bootstrapping",
author = "Brakerski, Zvika and Gentry, Craig and Vaikuntanathan, Vinod",
journal = "ACM Transactions on Computation Theory",
publisher = "Association for Computing Machinery (ACM)",
volume =  6,
number =  3,
pages = "1--36",
year =  2014
}

@ARTICLE{Fan2012-jt,
title = "Somewhat Practical Fully Homomorphic Encryption",
author = "Fan, Junfeng and Vercauteren, Frederik",
journal = "Cryptology ePrint Archive",
year =  2012
}

@ARTICLE{Gentry2013-eq,
title = "Homomorphic encryption from learning with errors: Conceptually-simpler, asymptotically-faster, attribute-based",
author = "Gentry, Craig and Sahai, A and Waters, Brent",
journal = "Annual International Cryptology Conference",
pages =  "75--92",
year =  2013
}

@ARTICLE{Cheon2017-pd,
title = "Homomorphic encryption for arithmetic of approximate numbers",
author = "Cheon, Jung Hee and Kim, Andrey and Kim, Miran and Song, Yongsoo",
journal = "International Conference on the Theory and Application of Cryptology and Information Security",
pages = "409--437",
year =  2017
}

@ARTICLE{Darup2021-qq,
title = "Encrypted Control for Networked Systems: An Illustrative Introduction and Current Challenges",
author = "Darup, Moritz Schulze and Alexandru, Andreea B and Quevedo, Daniel E and Pappas, George J",
journal = "IEEE Control Systems Magazine",
publisher = "IEEE",
volume =  41,
number =  3,
pages = "58--78",
year =  2021
}

@ARTICLE{Kim2022-ck,
title = "Comparison of encrypted control approaches and tutorial on dynamic systems using Learning With Errors-based homomorphic encryption",
author = "Kim, Junsoo and Kim, Dongwoo and Song, Yongsoo and Shim, Hyungbo and Sandberg, Henrik and Johansson, Karl H",
journal = "Annual Reviews in Control",
publisher = "Elsevier BV",
volume =  54,
pages = "200--218",
year =  2022
}

@ARTICLE{Schluter2023-ia,
title = "A brief survey on encrypted control: From the first to the second generation and beyond",
author = "Schl{\"{u}}ter, Nils and Binfet, Philipp and Schulze Darup, Moritz",
journal = "Annual Reviews in Control",
publisher = "Elsevier BV",
volume =  56,
pages =  100913,
year =  2023
}

@INPROCEEDINGS{Kogiso2015-go,
title = "Cyber-security enhancement of networked control systems using homomorphic encryption",
author = "Kogiso, Kiminao and Fujita, Takahiro",
booktitle = "IEEE Conference on Decision and Control",
pages = "6836--6843",
year =  2015
}

@ARTICLE{Kim2016-lc,
title = "Encrypting Controller using Fully Homomorphic Encryption for Security of Cyber-Physical Systems",
author = "Kim, Junsoo and Lee, Chanhwa and Shim, Hyungbo and Cheon, Jung Hee and Kim, Andrey and Kim, Miran and Song, Yongsoo",
journal = "IFAC-PapersOnLine",
publisher = "Elsevier BV",
volume =  49,
number =  22,
pages = "175--180",
year =  2016
}

@ARTICLE{Murguia2020-ov,
title = "Secure and private implementation of dynamic controllers using semi-homomorphic encryption",
author = "Murguia, Carlos and Farokhi, Farhad and Shames, Iman",
journal = "IEEE Transactions on Automatic Control",
volume =  65,
number =  9,
pages = "3950--3957",
year =  2020
}

@INPROCEEDINGS{Schluter2021-ek,
title = "Encrypted dynamic control with unlimited operating time via {FIR} filters",
author = "Schl{\"{u}}ter, Nils and Neuhaus, Matthias and Darup, Moritz Schulze",
booktitle = "European Control Conference",
pages = "952--957",
year =  2021
}

@INPROCEEDINGS{Darup2019-ou,
title = "Encrypted cloud-based control using secret sharing with one-time pads",
author = "Darup, Moritz Schulze and Jager, Tibor",
booktitle = "IEEE Conference on Decision and Control",
pages = "7215--7221",
year =  2019
}

@ARTICLE{Darup2020-cm,
title = "Encrypted polynomial control based on tailored two-party computation",
author = "Darup, Moritz Schulze",
journal = "International Journal of Robust and Nonlinear Control",
volume =  30,
number =  11,
pages = "4168--4187",
year =  2020
}

@INPROCEEDINGS{Schlor2021-am,
title = "Multi-party computation enables secure polynomial control based solely on secret-sharing",
author = "Schlor, Sebastian and Hertneck, Michael and Wildhagen, Stefan and Allg{\"{o}}wer, Frank",
booktitle = "IEEE Conference on Decision and Control",
pages = "4882--4887",
year =  2021
}

@INPROCEEDINGS{Tjell2021-ar,
title = "Secure learning-based {MPC} via garbled circuit",
author = "Tjell, Katrine and Schl{\"{u}}ter, Nils and Binfet, Philipp and Darup, Moritz Schulze",
booktitle = "IEEE Conference on Decision and Control",
pages = "4907--4914",
year =  2021
}

@INPROCEEDINGS{Alexandru2018-jv,
title = "Cloud-Based {MPC} with Encrypted Data",
author = "Alexandru, Andreea B and Morari, Manfred and Pappas, George J",
booktitle = "IEEE Conference on Decision and Control",
pages = "5014--5019",
year =  2018
}

@INCOLLECTION{Ohata2020-py,
title = "Communication-efficient (client-aided) secure two-party protocols and its application",
author = "Ohata, Satsuya and Nuida, Koji",
booktitle = "Financial Cryptography and Data Security",
publisher = "Springer International Publishing",
pages = "369--385",
year =  2020
}

@ARTICLE{Satriawan2023-gz,
title = "Conceptual review on number theoretic transform and comprehensive review on its implementations",
author = "Satriawan, Ardianto and Syafalni, Infall and Mareta, Rella and Anshori, Isa and Shalannanda, Wervyan and Barra, Aleams",
journal = "IEEE Access",
publisher = "Institute of Electrical and Electronics Engineers (IEEE)",
volume =  11,
pages = "70288--70316",
year =  2023
}

@INPROCEEDINGS{Cheon2018-vr,
title = "Need for controllers having integer coefficients in homomorphically encrypted dynamic system",
author = "Cheon, Jung Hee and Han, Kyoohyung and Kim, Hyuntae and Kim, Junsoo and Shim, Hyungbo",
booktitle = "IEEE Conference on Decision and Control",
pages = "5020--5025",
year =  2018
}

@ARTICLE{Han2018-pm,
title = "Privacy in control and dynamical systems",
author = "Han, Shuo and Pappas, George J",
journal = "Annual Review of Control, Robotics, and Autonomous Systems",
publisher = "Annual Reviews",
volume =  1,
pages = "309--332",
year =  2018
}

\begin{IEEEbiography}[{\includegraphics[width=1in,height=1.25in,clip,keepaspectratio]{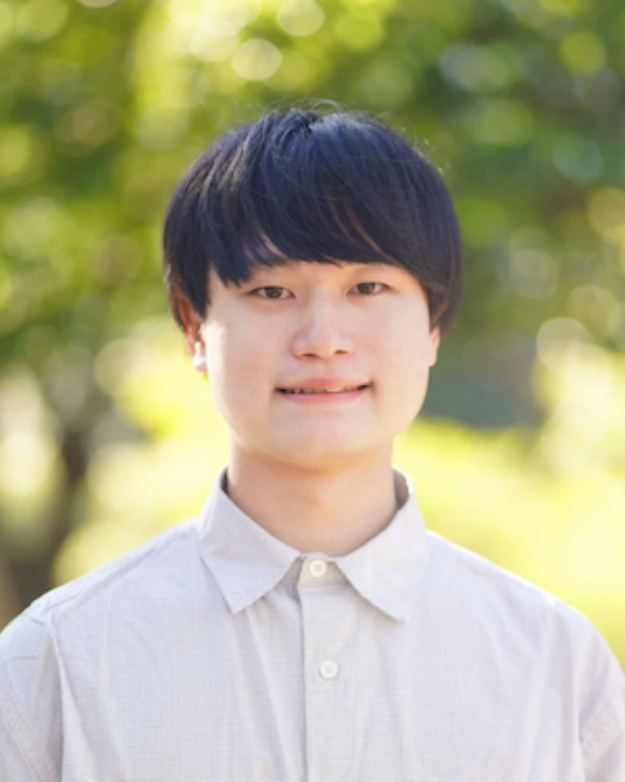}}]
{Kaoru Teranishi} (Member, IEEE) received his B.S. degree in electronic and mechanical engineering from the National Institute of Technology, Ishikawa College, Ishikawa, Japan, in 2019, and M.S. and Ph.D. degrees in mechanical and intelligent systems engineering from the University of Electro-Communications, Tokyo, Japan, in 2021 and 2024, respectively.
From October 2019 to September 2020, he was a Visiting Scholar at the Georgia Institute of Technology, GA, USA. 
From April 2021 to March 2024, he was a Research Fellow of the Japan Society for the Promotion of Science (JSPS), Tokyo, Japan.
From June 2024 to January 2025, he was a Research Affiliate Postdoctoral at the University of Texas at Austin, TX, USA.
He is currently a JSPS Overseas Research Fellow and a Visiting Scholar at Purdue University, IN, USA.
His research interests include control theory and cryptography for control systems security.
\end{IEEEbiography}

\begin{IEEEbiography}[{\includegraphics[width=1in,height=1.25in,clip,keepaspectratio]{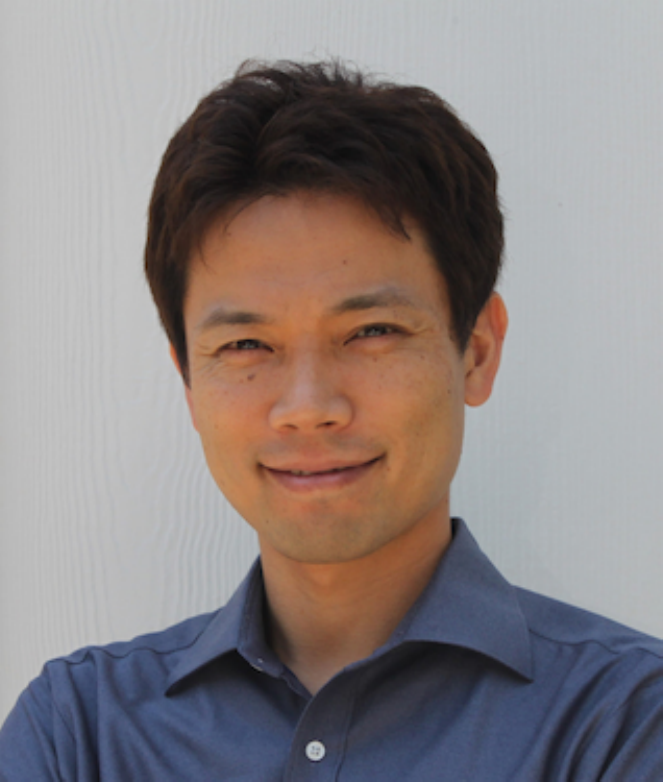}}]
{Takashi Tanaka} (Senior Member, IEEE) received the B.S. degree in Aerospace Engineering from the University of Tokyo, Tokyo, Japan, in 2006, and the M.S. and Ph.D. degrees in Aerospace Engineering (Automatic Control) from the University of Illinois at Urbana-Champaign, Urbana, IL, USA, in 2009 and 2012, respectively.
He was a Postdoctoral Associate with the Laboratory for Information and Decision Systems, Massachusetts Institute of Technology, Cambridge, MA, USA, from 2012 to 2015, and a Postdoctoral Researcher at KTH Royal Institute of Technology, Stockholm, Sweden, from 2015 to 2017.
He was an Assistant Professor in the Department of Aerospace Engineering and Engineering Mechanics at the University of Texas at Austin, Austin, TX, USA between 2017 and 2024, where he was an Associate Professor in 2024.
Since 2025, he has been an Associate Professor of Aeronautics and Astronautics and of Electrical and Computer Engineering at Purdue University, West Lafayette, IN, USA.
Dr. Tanaka was the recipient of the DARPA Young Faculty Award, the AFOSR Young Investigator Program Award, and the NSF Career Award.
\end{IEEEbiography}

\end{document}